\documentclass[twoside]{IEEEtran}

\ifCLASSINFOpdf
 \usepackage[pdftex]{graphicx}
 \DeclareGraphicsExtensions{.pdf,.jpeg,.png,.eps}
\else
 \usepackage[dvips]{graphicx}
 \DeclareGraphicsExtensions{.eps}
\fi
\usepackage[cmex10]{amsmath}
\usepackage{cases}
\usepackage[normalem]{ulem}
\usepackage{array}
\newcolumntype{C}[1]{>{\centering\arraybackslash}p{#1}} 

\usepackage{amssymb,amsthm}
\usepackage{mathtools, cuted}
\usepackage{lipsum, color, colortbl}
\usepackage[utf8]{inputenc}
\usepackage{cite}
\makeatletter
\newcommand{\citecomment}[2][]{\citen{#2}#1\citevar}
\newcommand{\citeone}[1]{\citecomment{#1}}
\newcommand{\citetwo}[2][]{\citecomment[,~#1]{#2}}
\newcommand{\citevar}{\@ifnextchar\bgroup{;~\citeone}{\@ifnextchar[{;~\citetwo}{]}}}
\newcommand{\citefirst}{\@ifnextchar\bgroup{\citeone}{\@ifnextchar[{\citetwo}{]}}}

\makeatother
\usepackage{empheq} 
\usepackage[pagewise]{lineno}

\usepackage[normalem]{ulem}
\usepackage{multirow}
\usepackage{array,booktabs}
\newcolumntype{P}[1]{>{\centering\arraybackslash}p{#1}}
\definecolor{Gray}{gray}{0.9}
\usepackage[tight,footnotesize]{subfigure}
\usepackage{textcomp}
\usepackage{gensymb}
\usepackage{flushend}

\newcommand{\argmax}{\operatornamewithlimits{arg\,max}}
\newcommand{\abs}[1]{\left\lvert{#1}\right\rvert}
\newcommand{\norm}[1]{\left\lVert{#1}\right\rVert}

\newcommand{\stkout}[1]{\ifmmode\text{\sout{\ensuremath{#1}}}\else\sout{#1}\fi}

\newtheorem{theorem}{Theorem}

\newtheorem{cor}{Corollary}
\newtheorem{remark}{Remark}
\newtheorem{definition}{Definition}
\newtheorem{problem}{Problem}

\begin{document}

\title{Achieving Positive Covert Capacity over MIMO AWGN Channels}
%
%
%

\author{Ahmed~Bendary,~\IEEEmembership{Student~Member,~IEEE,}
 Amr~Abdelaziz,~
 and~C.~Emre~Koksal,~\IEEEmembership{Senior~Member,~IEEE}

\thanks{This work was presented in part in 2020 IEEE Conference on Communications and Network Security (IEEE CNS 2020).}
\thanks{This work was in part supported by grants NSF CNS 1618566 and 1514260.}

\thanks{Ahmed Bendary and C. Emre Koksal are with the Department of Electrical and Computer Engineering, The Ohio State University, Columbus, OH 43210, USA, e-mail: {bendary.1, koksal.2@osu.edu}.}

\thanks{Amr Abdelaziz is with the Department of Communication Engineering, The Military Technical College, Cairo, Egypt, e-mail: {amrashry@mtc.edu.eg}.}
}
\maketitle
\begin{abstract}

We consider covert communication, i.e., hiding the presence of communication from an adversary for multiple-input multiple-output (MIMO) additive white Gaussian noise (AWGN) channels. We characterize the maximum covert coding rate under a variety of settings, including different regimes where either the number of transmit antennas or the blocklength is scaled up. We show that a non-zero covert capacity can be achieved in the massive MIMO regime in which the number of transmit antennas scales up but under specific conditions. Under such conditions, we show that the covert capacity of MIMO AWGN channels converges the capacity of MIMO AWGN channels. Furthermore, we derive the order-optimal scaling of the number of covert bits 
in the regime where the covert capacity is zero. We provide an insightful comparative analysis of different cases in which secrecy and energy-undetectability constraints are imposed separately or jointly.

\end{abstract}

\IEEEpeerreviewmaketitle
\section{Introduction}

\IEEEPARstart{W}{ireless} communication is prone to eavesdropping, and hence, cryptographic techniques are utilized to achieve secure wireless communication. However, in many situations, it is required that the entire communication session remains undetectable to protect the privacy of legitimate parties. This is fundamentally different from information secrecy (IS) that brings into action a set of new challenges that need to be addressed at the physical layer. Existing techniques such as Steganography, which hides information in non-secret files, and tools such as Tor, which provides anonymity for Internet users, are deployed in the application layer.

Recently, there has been a growing interest in low probability of detection (LPD) communication, referred to as covert communication. In a covert communication scenario, a legitimate transmitter, Alice, aims to send a message to a legitimate receiver, Bob, while an illegitimate receiver, Willie, tries to detect the existence of the ongoing communication. While Willie eavesdrops on the channel, Alice's objective is to reliably communicate the message to Bob and at the same time guarantee that Willie’s optimal detector is as inefficient as random guessing. One of the challenging exercises in secure communication is the reliable exchange of bits between two parties, while remaining covert. This problem has many unique challenges and received significant interest lately. One of the main findings is that the square root of the blocklength is an asymptotically-tight upper bound on the information bits that can be transmitted covertly over additive white Gaussian noise (AWGN) channels. Thus, the covert capacity is zero.


Hero had earlier considered a version of covert communication problem over multiple-input multiple-output (MIMO) channels in \cite{hero2003secure}, where the constraint for LPD is different from the typical LPD constraint of the probability of detection error being close to 1. Inspired by Steganography \cite{Fridrich_steganography1}, it has been established that the maximum number of bits that can be transmitted reliably with LPD is $\mathcal{O}(\sqrt{n})$ bits for both discrete memoryless channels (DMC) \cite{wang2016fundamental,Bloch_resolvability} and single-input single-output (SISO) AWGN channels \cite{wang2016fundamental,Bloch_resolvability,bash2013limits} where $n$ is the blocklength. Besides, in MIMO AWGN channels, when the channel state information (CSI) of the illegitimate receiver is only known to have a bounded spectral norm, the maximum number of covert bits is $\mathcal{O}(\sqrt{n})$ bits \cite{amr}. Therefore, the covert capacity is zero. However, in more practical settings, when the illegitimate receiver has uncertainty about the channels' conditions or the transmit time \cite{Survey2015,channel_uncertain,positive_rate2,Wang_noncausal,Magazine_Hanly}, a positive covert capacity can be achieved. Note that, a typical assumption to achieve covert communication is that a secret is shared between the legitimate parties \cite{bash2013limits,wang2016fundamental}. The shared secret can be the codebook itself such that only one codebook is used to transmit one message \cite{bash2013limits}, or a secret key that is added only once to the codeword \cite{Bloch_resolvability,wang2016fundamental}. Throughout this paper, we use the term {\it secret codebook}, i.e., the codebook is kept secret from the illegitimate receiver, to indicate that a secret is shared between the legitimate parties. It is shown that $\mathcal{O}(\sqrt{n})$ bits can be transmitted reliably and covertly using $\mathcal{O}(\sqrt{n})$ pre-shared bits \cite{Bloch_resolvability}. Nonetheless, covert communication can be achieved even when there is no shared secret under the condition that the illegitimate receiver's channel is ``noisier'' than the legitimate receiver's channel \cite{constant_rate,Bloch_resolvability,Survey2015}.

In covert communication, the problem of detection of Alice by Willie is posed as a stochastic signal detection problem. In the AWGN channel scenario, with the assumption that the codebook is kept secret from Willie, the optimal test that minimizes the probability of detection error is the likelihood ratio test \cite[Theorem 13.1.1]{lehmann2006testing}, which reduces to energy detection. That is, with a {\it secret codebook}, the LPD constraint is an energy-undetectability constraint. One may think that communication with the energy-undetectability constraint automatically implies communication with IS. Interestingly, however, an illegitimate receiver can decode the received message even if the existence of communication may not be energy detected. 

More recently, MIMO wireless communication is proposed to improve wireless security as well as to offer diversity and multiplexing gains. Accordingly, extensive work on communication with IS established the secrecy capacity of MIMO wiretap channels \cite{MMSE_approach,Khisti_MIMO,Hassibi_secrecy_capacity}. Further, the massive MIMO antenna array achieves high gain and directivity that offers LPD and high confidentiality against eavesdropping attacks \cite{Ozan,MIMO_against_eavesdropping}. However, MIMO channels have not been thoroughly understood in the setting of covert communication. Especially in the massive MIMO limit, it is unclear whether the square-root law still holds, as the number of transmit antennas is scaled up with the blocklength. 


In this paper, our objective is to show the effect of utilizing MIMO channels to achieve a non-zero covert capacity in contrast to the SISO case, in which the covert capacity is zero. We show that scaling up the number of transmit antennas or exploiting the null-space of the illegitimate receiver's channel achieves a non-zero covert capacity. To investigate this problem, we first start by characterizing the fundamental limits of covert communication over MIMO AWGN channels. However, without scaling up the number of transmit antennas or the existence of a null-space, the covert capacity of MIMO AWGN channels is zero similar to what is reported in the literature for DMC and SISO AWGN channels. Parallel to our work, the authors in \cite{Covert_MIMO} investigated covert communication over MIMO AWGN channels under total variational distance as the covertness metric where binary phase-shift keying and Gaussian signaling are shown to be optimal. However, therein, the number of transmit antennas is not more than the number of receive antennas
, and thus, achieving a positive covert capacity is not investigated. We investigate the fundamental limits with a {\it secret codebook} in Section \ref{sec:covert} and without a {\it secret codebook} in Section \ref{sec:covert and secure}, then, we show how to achieve a positive covert capacity in Section \ref{sec:asymptotic}. 

In particular, we start by providing the system model of MIMO AWGN channels and preliminary definitions relevant to covert communication. We explain the hypothesis statistical testing problem that leads to the energy-undetectability constraint in terms of the n-letter Kullback-Leibler (KL) divergence. Then, we state the associated problems that we investigate in this paper. We focus on two different regimes: 1) the scaling with the blocklength, and 2) the scaling with the number of transmit antennas. 

First, we consider the problem of maximizing the covert information that can be transmitted reliably over MIMO AWGN channels. We present a single-letter characterization of how the covert information scales with the blocklength. We derive an exact order-optimal expression for the scaling. When the covert capacity is zero, only a multiplicative constant of the square root of the blocklength, $\mathcal{O}(\sqrt{n})$, bits can be transmitted covertly and reliably over such channels. This result holds either with or without a secret codebook. In the case that the codebook is not kept secret, we incorporate different notions of IS into the problem (such as weak, strong, and effective secrecy) and consider one of them simultaneously with the energy-undetectability constraint. Besides, we deduce the scaling laws for special MIMO AWGN channels, which models different scenarios, such as well-conditioned MIMO AWGN channels, unit-rank MIMO AWGN channels as well as the compound MIMO AWGN channels when the CSI of the illegitimate receiver is only known to have a bounded spectral norm. Further, we provide a comparative discussion to explain the relationship between communication with IS, covert communication, and energy-undetectable communication. 

Second, we show that the covert capacity does not need to be zero for a sufficiently large number of transmit antennas. Note that we make only mild assumptions on the CSI of the adversary who is trying to detect the communication session. This is in contrast to the existing studies, many of which assume the presence of global CSI to achieve covert communication. Also, we provide a lower bound on the number of transmit antennas that achieves a predefined target probability of detection. Finally, we provide numerical results to illustrate the behavior of covert communication when either the number of transmit antennas or the blocklength scales up. 

The contributions of this paper can be summarized as follows: 
\begin{itemize}
\item With a secret codebook and a known CSI of the illegitimate receiver, we derive order-optimal scaling of the maximum number of covert bits that can be transmitted reliably over MIMO AWGN channels. We show that without the existence of a null-space of the illegitimate receiver's channel, the covert capacity is zero but $\mathcal{O}(\sqrt{n})$ bits can be transmitted covertly and reliably. Thus, without a null-space, the square-root law still holds for MIMO channels as in the literature for SISO channels \cite{bash2013limits,Bloch_resolvability,wang2016fundamental}.
\item Without a secret codebook but with a known CSI of the illegitimate receiver, we derive order-optimal scaling of the maximum number of covert bits that can be transmitted reliably with IS over MIMO AWGN channels. We show that without the existence of a null-space of the illegitimate receiver's channel, the covert capacity is zero but $\mathcal{O}(\sqrt{n})$ bits can be transmitted covertly and reliably with IS, which coincides with \cite{Bloch_resolvability} and Model 1 in \cite{constant_rate}. 
\item With a secret codebook but without the CSI of the illegitimate receiver, we provide conditions under which the covert capacity of MIMO AWGN channels converges to the capacity of MIMO AWGN channels (without the KL constraint) for a sufficiently large number of transmit antennas, i.e., the square-root law of covert communication can be overcome for MIMO AWGN Channels. This result is parallel to the findings for IS in \cite{Ozan,MIMO_against_eavesdropping}, which deal with eavesdropping attacks. Besides, we provide a lower bound on the number of transmit antennas that achieves a predefined target probability of detection.
\end{itemize}

{\em Organization:} The paper is organized as follows. The system model, preliminary definitions, and problem formulation are presented in Section \ref{sec:model}. The scaling of the maximum number of covert bits with a secret codebook is provided in Section \ref{sec:covert}. The scaling of the maximum number of covert bits without a secret codebook but with IS is provided in Section \ref{sec:covert and secure}. In Section \ref{sec:asymptotic}, the asymptotics of covert communication with the number of transmit antennas is developed. Furthermore, we provide numerical examples in Section \ref{numerical} and the paper is concluded in Section \ref{sec:conclusion}.

{\em Notations:} In the rest of this paper, $\max\,(0,\cdot)$ is denoted $[\cdot]^+$, the transpose and the conjugate transpose are denoted by $(\cdot)^{T}$, and $(\cdot)^{ \dag}$, respectively. The identity matrix of a size $N$ is denoted by $\mathbf{I}_N$, the trace and the determinant of a matrix ${\mathbf{A}}$ are denoted by $\mathbf{tr}(\mathbf{A})$ and $\abs{\mathbf{A}}$, respectively. A diagonal matrix with diagonal elements $(a_1, a_2,\dots , a_m)$ is denoted by $\mathbf{diag}(a_1, a_2,\dots , a_m)$. When $\mathbf{A}-\mathbf{B}$ is a positive semi-definite, we write $\mathbf{A}\succeq\mathbf{B}$. The distribution of a circularly symmetric complex Gaussian random vector $\mathbf{X}$ with a mean vector $\mathbf{\mu}$ and a covariance matrix $\mathbf{Q}$ is denoted by $\mathcal{CN}(\mathbf{\mu},\mathbf{Q})$. Moreover, the expectation and the variance of a random vector $\mathbf{X}$ are denoted by $\mathbb{E}[\mathbf{X}]$ and $\mathbf{var}[\mathbf{X}]$, respectively. The entropy of a discrete random variable $X$ is $H(X)$ and the differential entropy of a continuous random variable $X$ is $h(X)$, while the mutual information between two random variables $X$ and $Y$ is denoted by $\mathrm{I}(X;Y)$. Further, $\varliminf$, and $\norm{.}_{op}$ are the limit inferior, and the operator (spectral) norm of a matrix $\mathbf{A}$, i.e., the maximum eigenvalue of $\mathbf{A}$, respectively. An asymptotically tight upper bound on a function $f(n)$ is denoted by $\mathcal{O}(g(n))$, i.e., there exist constants $m$ and $n_0>0$ such that $0 \leq f(n) \leq m\,g(n)$ for all $n>n_0$. The KL divergence between two distributions is denoted as $\mathcal{D}(\mathbb{P}_{1} \parallel \mathbb{P}_{0})\,\triangleq\mathbb{E}_{\mathbb{P}_{1}} \left[\log f_{1} - \log f_{0}\right],$ where $f_{1}$ and $f_{0}$ are the density functions of $\mathbb{P}_{1}$ and $\mathbb{P}_{0}$, respectively. The logarithm function is the natural logarithm and information is measured in nats. 

\section{System Model}\label{sec:model}
\begin{figure}[h]
 \centering
 \includegraphics[scale=.58]{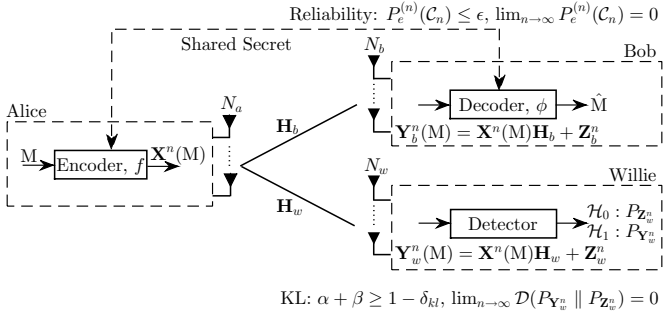}
 \caption{Covert communication over MIMO AWGN channels with a secret codebook.\label{fig:model}}
\end{figure}

In this paper, we consider MIMO AWGN channels. In this model, a legitimate transmitter, Alice, equipped with $N_a$ antennas, aims to send a covert message to a legitimate receiver, Bob, equipped with $N_b$ antennas, while an illegitimate receiver, Willie, equipped with $N_w$ antennas tries to detect the existence of the ongoing communication. The received signals at Bob and Willie, respectively, can be written as:
\begin{equation}\begin{aligned}
\label{eq:inout_security}
\mathbf{Y}_b&=\mathbf{H}_b\,\mathbf{X} + \mathbf{Z}_b, \\
\mathbf{Y}_w&=\mathbf{H}_w\,\mathbf{X} + \mathbf{Z}_w,
\end{aligned} \end{equation} 
where $\mathbf{X} \in \mathbb{C}^{N_a}$ is the transmitted signal with an average power constraint $\mathbf{tr}(\mathbf{Q})\leq P$, $\mathbf{Q}=\mathbb{E}[\mathbf{X}\mathbf{X}^{\dag}]$, $\mathbf{Z}_b \in \mathbb{C}^{N_b}$ and $\mathbf{Z}_w \in \mathbb{C}^{N_w}$ are independent and identically-distributed (i.i.d.) random vectors at Bob and Willie, respectively, $\mathbf{Z}_u\sim \mathcal{CN}(\mathbf{0},{\bf\Sigma}_b=\sigma_u^2\, \mathbf{I}_{N_u})$ and $ u\in\left\{b,w\right\}$. Moreover, ${\mathbf{H}_b} \in \mathbb{C}^{N_b \times N_a}$, denoted by Bob's channel, and $\mathbf{H}_w \in \mathbb{C}^{N_w \times N_a}$, denoted by Willie's channel, are the channels' coefficient matrices of Alice-Bob's and Alice-Willie's channels, respectively. Throughout this paper, $\mathbf{H}_b$ and $\mathbf{H}_w$ are assumed to be deterministic. More precisely, in MIMO AWGN channels, the channel matrices are time-invariant throughout the communication session \cite{tse2005fundamentals}. We start by assuming that both are known. Then, in Section \ref{sec:asymptotic}, we relax this assumption. Across $n$ channel uses, the transmitted sequence is defined by ${\bf X}^n\triangleq\left\{{\bf X}_1,{\bf X}_2,\dots,{\bf X}_n\right\}$ and ${\bf Y}^n_b$, ${\bf Y}^n_w$, ${\bf Z}^n_b$, and ${\bf Z}^n_w$ are defined similarly. The channel is memoryless and stationary such that to send a message, Alice sends ${\bf X}^n$ across $n$ uses. Bob and Willie receive ${\bf Y}_{b}^n$ and ${\bf Y}_{w}^n$, respectively, according to the following conditional density function:
\begin{equation}\label{memoryless}
 f_{{\bf Y}_{u}^n|{\bf X}^n}({\bf y}_u^n|{\bf x}^n)=\prod_{i=1}^n f_{{\bf Y}_{u}|{\bf X}}({\bf y}_{{u},i}|{\bf x}_i),\ u\in\left\{b,w\right\}.
\end{equation}
Next, we formulate problems that are addressed in this paper based on some preliminary definitions. 
\begin{definition}\label{code_definition}
A $(2^{{\lceil nR\rceil}},n,\epsilon)$-code consists of:
\begin{enumerate}
\item A uniformly distributed message set, $\mathcal{M}_n\triangleq[1:2^{\lceil nR\rceil}]$.
\item An encoder, $f:\mathcal{M}_n\mapsto \mathbb{C}^{n}$, that assigns a codeword (random vector), ${x}^n(m)$, to each message, $m\in\mathcal{M}_n$, under an average transmit power constraint $P$ on every codeword, ${\bf x}^n(m)\triangleq\left\{{\bf x}_1(m),\ \dots,{\bf x}_n(m)\right\}$, as follows:
\begin{equation}
\sum_{i=1}^n{\bf x}^{\dag}_i(m){\bf x}_i(m)\leq n\,P.
\end{equation}
\item A decoder, $\phi:\mathbb{C}^{n} \mapsto \mathcal{M}_n\cup{\left\{e\right\}}$, that assigns an estimate, $\hat{m}\in\mathcal{M}_n$, or an error message, $e$, to each received sequence, ${y}^{n}$.
\item The encoder-decoder pair satisfies the average probability of decoding error constraint: $P_e^{(n)}(\mathcal{C}_n)=\mathbb{P}\left[M\neq\hat{M}\right]\leq \epsilon$ where $\mathcal{C}_n$ is the codebook associated with the code.
\end{enumerate} 
\end{definition}

\begin{definition}[{\bf A secret codebook}]\label{secret codebook}
The secret codebook denotes either 1) a codebook that is kept secret and only one codebook is used to transmit one message \cite{bash2013limits} or 2) one-time padding the codeword with a secret of a sufficient length \cite{wang2016fundamental}.
\end{definition}

In the following, we discuss the energy-undetectability constraint for covert communication. This constraint is adopted in the literature and in this paper as well. {\bf Willie's objective} is to detect whether there is communication between Alice and Bob or not. When the codebook is kept secret, the only information at Willie becomes unknown random instead of being in a finite set of discrete possibilities. Thus, Willie cannot employ coherent detection. Specifically, with a secret codebook and known channel, $\mathbf{H}_w$, the problem reduces to stochastic signal detection using the received observations, ${\mathbf{Y}^n_w}={\mathbf{y}^n_w}$. Therefore, Willie employs statistical hypothesis testing, which is now optimal to detect the presence of communication. ``Alice is transmitting" is the true hypothesis, $\mathcal{H}_1$, where the probability distribution of Willie's observation is denoted by $\mathbb{P}_{\mathbf{Y}_w^n}$. ``Alice is not transmitting" is the null hypothesis, $\mathcal{H}_0$, where the probability distribution of Willie's observation is denoted by $\mathbb{P}_{\mathbf{Z}_w^n}$. The detection problem is given as follows:
\begin{equation}\begin{aligned}\label{HST}
\mathcal{H}_0:&\ \mathbf{Y}_w^n={\bf Z}_w^n, \\
\mathcal{H}_1:&\ \mathbf{Y}_w^n=\mathbf{H}_w\,{\bf X}^n+{\bf Z}_w^n.
\end{aligned} \end{equation} 

The probability of rejecting $\mathcal{H}_0$ when it is true, type \textbf{I} error, is denoted by $\alpha$, while the probability of rejecting $\mathcal{H}_1$ when it is true, type \textbf{II} error, is denoted by $\beta$. The optimal test that minimizes the sum of the detection error probabilities, i.e., minimizes $\alpha+\beta$, over all possible tests, is the likelihood ratio test \cite[Theorem 13.1.1]{lehmann2006testing}, which reduces to energy detection, and is given by:
\begin{equation}
\parallel\mathbf{y}^n_w\parallel^2 \underset{\mathcal{H}_0}{\overset{\mathcal{H}_1}{\gtreqqless}}{\tau},
\end{equation} 
where $\tau$ is the optimum threshold that minimizes $\alpha+\beta$. More precisely, with a secret codebook and known channel, $\mathbf{H}_w$, Willie cannot employ a coherent detection to decode reliably. Besides, the optimal detector is an energy detector with a threshold, $\tau$, that minimizes $\alpha+\beta$.

In contrast, {\bf Alice's objective} is to guarantee that Willie’s test is as inefficient as a blind test for which $\alpha+\beta=1$ \cite{Bloch_resolvability}. The minimum sum of the detection error probabilities by the likelihood ratio test is given by \cite{lehmann2006testing}:
\begin{equation}\begin{aligned}
\alpha + \beta = 1 - \mathcal{V}(\mathbb{P}_{\mathbf{Y}_w^n},\mathbb{P}_{\mathbf{Z}_w^n}),
\end{aligned}\end{equation}
where $\mathcal{V}(\mathbb{P}_{\mathbf{Y}_w^n},\mathbb{P}_{\mathbf{Z}_w^n})\triangleq \frac{1}{2}\parallel \mathbb{P}_{\mathbf{Y}_w^n}-\mathbb{P}_{\mathbf{Z}_w^n}{\parallel}_1$ is the variational distance between two distributions, $\mathbb{P}_{\mathbf{Y}_w^n}$ and $\mathbb{P}_{\mathbf{Z}_w^n}$. Further, to lower bound $\alpha+\beta$, we use the KL divergence to upper bound the variational distance using Pinsker’s inequality \cite{bash2013limits,bloch2019}:
\begin{equation}
\mathcal{V}(\mathbb{P}_{\mathbf{Y}_w^n},\mathbb{P}_{\mathbf{Z}_w^n}) \leq \sqrt{\frac{1}{2} \mathcal{D}(\mathbb{P}_{\mathbf{Y}_w^n} \parallel \mathbb{P}_{\mathbf{Z}_w^n})},
\end{equation} 
where 
\begin{equation}
\mathcal{D}(\mathbb{P}_{\mathbf{Y}_w^n} \parallel \mathbb{P}_{\mathbf{Z}_w^n})\triangleq\mathbb{E}_{\mathbb{P}_{\mathbf{Y}_w^n}} \left[\log f_{\mathbf{Y}_w^n}(\mathbf{Y}_w^n) - \log f_{\mathbf{Z}_w^n}(\mathbf{Z}_w^n)\right],\end{equation}
with $f_{\mathbf{Y}_w^n}(\mathbf{y}_w^n)$ and $f_{\mathbf{Z}_w^n}(\mathbf{z}_w^n)$ are the density functions of $\mathbb{P}_{\mathbf{Y}_w^n}$ and $\mathbb{P}_{\mathbf{Z}_w^n}$, respectively. Accordingly, to guarantee that Willie’s optimal detector is as inefficient as random guessing, and hence, achieve covert communication, Alice chooses the desired probability of detection, $\delta_{kl}$, to upper bound $\mathcal{D}(\mathbb{P}_{\mathbf{Y}_w^n} \parallel \mathbb{P}_{\mathbf{Z}_w^n})$. Consequently, the sum of the detection error probabilities is bounded as $\alpha + \beta \geq 1-\delta_{kl}$, regardless of the operating point on Willie’s ROC curve for any $\alpha$ \cite{bloch2019}. Hence, Alice can achieve covert communication using the optimal input covariance matrix that satisfies the following KL constraint:
\begin{equation}
\label{eq:metric}
\mathcal{D}(\mathbb{P}_{\mathbf{Y}_w^n} \parallel \mathbb{P}_{\mathbf{Z}_w^n}) \leq 2\,\delta_{kl}^2,\ \delta_{kl}\geq 0.
\end{equation}

It is worth mentioning that the KL divergence in the previous equation is more restrictive than the variational distance and than $\mathcal{D}(\mathbb{P}_{\mathbf{Z}_w^n} \parallel \mathbb{P}_{\mathbf{Y}_w^n})$ as well \cite{Gaussian_signaling}. Although both KL divergences suffice to bound the variational distance, the optimal signaling has not been derived yet when using $\mathcal{D}(\mathbb{P}_{\mathbf{Z}_w^n} \parallel \mathbb{P}_{\mathbf{Y}_w^n})$ as reported in \cite{Gaussian_signaling}. Note that covert communication under different metrics is investigated in \cite{bloch2019}.

\begin{remark}\label{remark KL necessary and suff}
With a secret codebook, although bounding the KL divergence (which is an energy-undetectability constraint) does not constrain the ROC curve tightly, minimizing $\mathcal{D}(\mathbb{P}_{\mathbf{Y}_w^n} \parallel \mathbb{P}_{\mathbf{Z}_w^n})$ is a sufficient condition to guarantee that Willie's test is ineffective \cite{bloch2019}. In other words, bounding the KL divergence by an arbitrary small $\delta_{kl}\geq 0$ bounds the total variational distance. This can be inferred from Pinsker's inequality. On the other hand, guaranteeing a zero variational distance asymptotically implies a zero KL divergence. This holds by the definitions of the total variational distance and the KL divergence. Besides, with a secret codebook, the KL asymptotic condition, $\lim_{n\rightarrow \infty}\mathcal{D}(\mathbb{P}_{\mathbf{Y}_w^n} \parallel \mathbb{P}_{\mathbf{Z}_w^n})= 0$, corresponds to the LPD notion in the literature \cite{wang2016fundamental,Bloch_resolvability,bash2013limits}. 
\end{remark}

\begin{definition}\label{def}
A $(2^{{\lceil nR\rceil}},n,\epsilon,\delta_{kl})$-code is a $(2^{{\lceil nR\rceil}},n,\epsilon)$-code that satisfies the KL constraint: $\mathcal{D}(\mathbb{P}_{\mathbf{Y}_w^n} \parallel \mathbb{P}_{\mathbf{Z}_w^n}) \leq 2\,\delta_{kl}^2$. 
\end{definition}

\begin{definition} \label{maximal_covert_coding}
Maximal covert coding rate:
\begin{equation}
{R}(n,\epsilon,\delta_{kl})\triangleq\sup \left\{R:\,\exists\ \mbox{a}\ (2^{\lceil nR\rceil},n,\epsilon,\delta_{kl})\mbox{-code}\right\}.
\end{equation}
\end{definition}


\begin{definition}\label{the_scaling_L}
The scaling, with the blocklength, of the maximum number of nats that can be transmitted covertly, while satisfying the average probability of decoding error constraint and the average power constraint, is defined as follows \cite{wang2016fundamental}:
\begin{equation}\label{eq:L}
L\triangleq \lim_{\epsilon \downarrow 0} \varliminf_{n \rightarrow \infty} \sqrt{\frac{n}{2\,\delta_{kl}^2}}{R}(n,\epsilon,\delta_{kl}).
\end{equation}
\end{definition}

\begin{problem} (With a secret codebook) Characterize the scaling, $L$, of the maximum number of covert nats over MIMO AWGN channels as defined in \eqref{eq:L}.
\end{problem}

This problem is investigated thoroughly in Section \ref{sec:covert} and results for special cases of MIMO AWGN channels are deduced therein. Note that, scaling variable $L$ is an asymptotic quantity and it does not depend on $n$. In Section \ref{sec:covert and secure}, we consider the case when there is no secret codebook. Hence, the KL constraint jointly with the IS constraint is considered. To formally state this problem, we provide related definitions and the problem statement therein. Interestingly, without the knowledge of Willie's CSI, we still can achieve a positive covert capacity by transmitting a very narrow beamwidth using massive MIMO. Thus, again, preventing the illegitimate receiver from detecting the ongoing communication. This problem is discussed in Section \ref{sec:asymptotic}.
\begin{problem} (With unknown CSI of Willie)
\begin{enumerate}
\item Can the square-root law of covert communication be overcome if the number of transmit antennas is scaled up? In particular, under what conditions can we achieve the following: $\lim_{N_a\rightarrow \infty}{R}(n,\epsilon,\delta_{kl})=\lim_{N_a\rightarrow \infty}{R}(n,\epsilon)$ for any given $\delta_{kl}\geq 0$? where ${R}(n,\epsilon)$ is the maximal coding rate (without the KL constraint).
\item What is the lower bound on the number of transmit antennas that satisfies a predefined target probability of detection?
\end{enumerate}
\end{problem}

The following table summarizes the main results and the underlying assumptions in this paper.

\begin{table*}[h]
\centering
\caption{Main results and assumptions.}
{\begin{tabular}{|c|c|c|c|c|c|c|c|c|}
 \hline
 {Covert capacity} & Quantity &{Willie's CSI} & {Codebook} & {Transmit antennas} & {Result}& {Section} \\ \hline
 {Positive/Zero} & Scaling expression &Known & {Secret} & {Arbitrary} & Theorem \ref{thm:Kn_covert}& { \ref{sec:covert}} \\ \hline
 {Positive/Zero} & Exact scaling &Known & {Secret} & {Arbitrary} & Theorem \ref{thm:Scaling}& { \ref{sec:covert}} \\ \hline
 {Positive/Zero} & Scaling expression &Known & {IS} & {Arbitrary} & Corollary \ref{thm:sec_covert}& { \ref{sec:covert and secure}} \\ \hline
{Positive/Zero} & Exact scaling &Known & {IS} & {Arbitrary} & Theorem \ref{thm:secrecy_Scaling}& { \ref{sec:covert and secure}} \\ \hline
{Positive} & Achievable rate &Known + Null-space& {Secret} & {Arbitrary} & Corollary \ref{positive_rate_cor}& { \ref{sec:covert}} \\ \hline
{Positive} & Coding rate of unit-rank&Unknown & {Secret} & {Sufficiently large} & Theorem \ref{prop:one}& { \ref{sec:asymptotic}} \\ \hline
{Positive} & Coding rate of multi-path&Unknown & {Secret} & {Sufficiently large} & Theorem \ref{prop:two}& { \ref{sec:asymptotic}} \\ \hline
 \end{tabular}}
\end{table*}
\section{Covert Communication over MIMO AWGN Channels with a Secret Codebook}\label{sec:covert}
In this section, the codebook is kept secret between Alice and Bob, and hence, the KL constraint is sufficient to achieve covert communication. Thus, covertness is investigated by analyzing the KL constraint. The following theorem extends the result in \cite{wang2016fundamental} for DMC to MIMO AWGN channels with infinite input and output alphabets. In this theorem, a single-letter characterization of the scaling is provided in terms of maximizing the capacity of MIMO AWGN channels subject to a single-letter KL divergence instead of the $n$-letter KL divergence.
\begin{theorem}
\label{thm:Kn_covert}
The scaling of the maximum number of covert nats that can be transmitted reliably over MIMO AWGN channels is given by:
\begin{equation}\label{eq:thm1}
\begin{aligned}
L= \varliminf_{n \rightarrow \infty} \sqrt{\frac{n}{2\,\delta_{kl}^2}}&\max_{\substack{\mathbf{Q}_n \succeq \mathbf{0} \\ 
\mathbf{tr}(\mathbf{Q}_n) \leq P}}
& & \mathrm{\log\abs{\frac{1}{\sigma_{\it b}^2}\,{\bf H}_{\it b}\,{\bf Q}_{\it n}\,{\bf H}_{\it b} ^\dag+{\bf I}_{N_{\it b}}}} \\
& \text{subject to:}
& &\mathcal{D}(\mathbb{P}_{\mathbf{Y}_w}\parallel\mathbb{P}_{\mathbf{Z}_w})\leq\frac{2\,\delta_{kl}^2}{n}.
\end{aligned}
\end{equation}

Moreover, the input distribution that maximizes the first-order approximation of the covert coding rate, while minimizing $\mathcal{D}(\mathbb{P}_{\mathbf{Y}_w}\parallel \mathbb{P}_{\mathbf{Z}_w})$, is the zero-mean circularly symmetric complex Gaussian distribution with a covariance matrix ${\bf Q}_n$.
\end{theorem}
 \begin{IEEEproof} The full proof can be found in Appendix \ref{App:proof_thm1}. We give a brief sketch as follows. First, we give a converse proof in terms of both the probability of decoding error and the KL divergence. We start with Fano's and data processing inequalities to derive an upper bound on the coding rate as in the proof of the channel coding theorem. Then, we derive a lower bound on the n-letter KL divergence in terms of the single-letter KL divergence similar to the approach in \cite{Hou_dependentRV} and \cite{wang2016fundamental}. Second, we provide an achievability proof by exploiting random coding according to a sequence of input distributions. However, when the achievable covert rate is zero, the achievability proof needs special treatment as pointed out in \cite{wang2016fundamental}. We generalize the achievability proof of Theorem 1 in \cite{wang2016fundamental} for DMC channels to prove the remaining part of the achievability of MIMO AWGN channels.
\end{IEEEproof}

We prove in Theorem \ref{thm:Kn_covert} that Gaussian signaling is optimal for covert communication over MIMO AWGN channels when using the KL divergence $\mathcal{D}(\mathbb{P}_{\mathbf{Y}_w^n} \parallel \mathbb{P}_{\mathbf{Z}_w^n})$, which is consistent with the recent findings in \cite{Gaussian_signaling} for SISO AWGN channels. However, Gaussian signaling is not optimal when using the KL divergence $\mathcal{D}(\mathbb{P}_{\mathbf{Z}_w^n}\parallel \mathbb{P}_{\mathbf{Y}_w^n})$ and further research is still needed to find the optimal signaling in general as reported in \cite{Gaussian_signaling}. We note that the scaling expression in Theorem \ref{thm:Kn_covert} does not hold in general when using the other KL divergence. Although this expression is achievable under a constraint on $\mathcal{D}(\mathbb{P}_{\mathbf{Z}_w^n} \parallel \mathbb{P}_{\mathbf{Y}_w^n})$, further research is needed to derive the optimal signaling that maximizes the mutual information (not necessarily the log determinant expression) using that KL divergence. Also, the scaling in Theorem \ref{thm:Kn_covert} and in other results, specifically, the square-root law, holds in general for covert communication over MIMO AWGN channels regardless of the used KL divergence. The only difference is the optimal power allocation under each metric.

\begin{remark}\label{remark nullspace}
The scaling expression in Theorem \ref{thm:Kn_covert} holds in general even when there exists a null-space of Willie's channel. However, when there exists a null-space, $L=\infty$. On the other hand, although it is not difficult to derive the achievability proof for this expression in general, it needs special treatment as pointed out in \cite{wang2016fundamental} since the achievable covert rate is zero in this case. Otherwise, in the case that a null-space exists, the KL constraint is not active and the optimal solution is reached at the maximum transmit power. Thus, $\mathcal{O}({n})$ covert nats can be transmitted reliably, and the covert capacity is an appropriate metric that is neither zero nor infinity. Therefore, the average power constraint is meaningful for MIMO AWGN channels, although it is inactive in SISO AWGN channels due to the KL constraint, which forces the average power to decay with the blocklength \cite{wang2016fundamental}. Also, Appendix \ref{App:power allocation} shows the power allocation for the optimization problem in Theorem \ref{thm:Kn_covert} in general, which includes different regimes when there exists a null-space or not.
\end{remark}

In the following, we do not impose any rank constraint on Bob's or Willie's channels. We use the generalized singular value decomposition \cite{Paige_GSVD} to give the single-letter KL divergence, which is derived in Appendix \ref{App:KL_divergence} as follows:
\begin{equation}\begin{aligned}
\mathcal{D}(\mathbb{P}_{\mathbf{Y}_w}\parallel\mathbb{P}_{\mathbf{Z}_w})=\sum_{i=1}^{N}\left[{\frac{q_i\,{\lambda}_{w,i}}{\sigma_{\it w}^2}}-\log\left(\frac{q_i\,{\lambda}_{w,i}}{\sigma_{\it w}^2}+1\right)\right],
\end{aligned} \end{equation}
where 
\begin{equation}
N=\text{rank}\left(
\begin{matrix}
{\bf H}_{\it w}\\
{\bf H}_{\it b}
\end{matrix}\right),
\end{equation}
$q_i$ is the transmit power in the $i^{th}$ eigen-direction, and ${\lambda}_{w,i}\geq 0,\ \forall i\in\left\{1,\dots,N\right\},$ are the eigenvalues of Willie's channel. We note that the optimal power allocation is a decreasing function of the blocklength. However, we omit this for notational convenience. Let $c_i$ be the normalized KL divergence in the $i^{th}$ eigen-direction, i.e., 
\begin{equation}\label{optimal_power_implicit}
c_i=\left[{\frac{q_i^*\,{\lambda}_{w,i}}{\sigma_{\it w}^2}-\log\left(\frac{q_i^*\,{\lambda}_{w,i}}{\sigma_{\it w}^2}+1\right)}\right]/{\mathcal{D}(\mathbb{P}_{\mathbf{Y}_w}\parallel\mathbb{P}_{\mathbf{Z}_w})},
\end{equation}
$\forall i\in\left\{1,\dots,N\right\}$, where $q_i^*$ is the optimal transmit power in the $i^{th}$ eigen-direction that achieves the maximum covert coding rate. Thus, $c_i$ is fixed and is independent of the blocklength $\forall i$, and $\sum_{i=1}^{N}c_i=1$. Let ${\lambda}_{b,i}\geq 0,\ \forall i\in\left\{1,\dots,N\right\},$ be the eigenvalues of Bob's channel. The following theorem extends the scaling law in \cite{wang2016fundamental} for SISO AWGN channels to MIMO AWGN channels.
\begin{theorem}
\label{thm:Scaling}
The scaling of the maximum number of nats that can be transmitted covertly and reliably over MIMO AWGN channels is given by:
\begin{equation}
L=\sum_{i=1}^{N}\frac{\sqrt{2\,c_i}\,{\sigma_w^2}\,{\lambda}_{\it b,i}}{\sigma_{\it b}^2\,{\lambda}_{w,i}},
\end{equation}
while ensuring that the illegitimate receiver's sum of the detection error probabilities is lower bounded by $(1-\delta_{kl})$.
\end{theorem}
\begin{IEEEproof}The full proof can be found in Appendix \ref{App:proof_thm2}. We give a brief sketch as follows. We give an upper bound on the scaling by exploiting a lower bound on the KL divergence and with the aid of the converse result of Theorem \ref{thm:Kn_covert}. Similarly, using the achievability of Theorem \ref{thm:Kn_covert}, we give a lower bound on the scaling by using an upper bound on the KL divergence. In the limit, both bounds coincide to the same constant.
\end{IEEEproof}

It is worth mentioning that Theorem \ref{thm:Scaling} extends the result in \cite{wang2016fundamental} for SISO AWGN channels to MIMO AWGN channels and both results coincide when $\sigma_b=\sigma_w$, $\lambda_{b,i}=\lambda_{w,i}$, $\forall i$, $N_a=N_b=1$ and using a real channel input. Then, the scaling is $L=1$ $\sqrt{\mbox{nat}}$. Without the existence of a null-space, the covert capacity is zero but the covert information is $\mathcal{O}(\sqrt{n})$ nats. Otherwise, $L=\infty$. Besides, Theorem \ref{thm:Scaling} shows the parameters that affect the covert coding rate such as: 1) the available eigen-directions between Alice and Bob, 2) Bob and Willie's channel power gain and distances from Alice, 3) the noise power density at Bob and Willie, and 4) the fraction of transmitted power in each eigen-direction.

\subsection{Scaling Laws of Specific MIMO AWGN Channels}
\subsubsection{{\bf Well-Conditioned MIMO AWGN Channels}}
In a rich scattering environment, the channel matrix is well-conditioned where the eigenvalues are independent and approximately identical, and hence, equal power allocation is optimal. The scaling, in this case, is given in the following corollary.
\begin{cor}
The scaling of the maximum number of covert nats that can be transmitted reliably over well-conditioned MIMO AWGN channels is given by:
\begin{equation}
L=\frac{\sqrt{2\,N}\,{\sigma_w^2}\,{\lambda}_{\it b}}{\sigma_{\it b}^2\,{\lambda}_{w}},
\end{equation}
while ensuring that the illegitimate receiver's sum of the detection error probabilities is lower bounded by $(1-\delta_{kl})$.
\end{cor}

\subsubsection{{\bf Bounded Spectral Norm of Willie's CSI}} (A compound channel setting) We consider a special case of the compound channel setting as in \cite[Sec. V]{schaefer2015secrecy} such that Bob's CSI is known and the CSI of Willie is unknown but lies in the set of channels with a bounded spectral norm,
\begin{equation}\label{spectral norm}
\mathcal{S}_w = \left\{\mathbf{H}_w: \frac{\norm{\mathbf{H}_w\,\mathbf{H}_w^{ \dag}}_{op}}{\sigma_{\it w}^2} \leq \hat{\lambda}_w\right\},
\end{equation}
where $\norm{\mathbf{H}_w\,\mathbf{H}_w^{ \dag}}_{op}$ represents the largest possible power gain of Willie's channel, i.e., the worst-case Willie's channel is isotropic. Hence, the set $\mathcal{S}_w$ incorporates all possible matrices, $\mathbf{H}_w$, such that $\frac{\mathbf{H}_w\,\mathbf{H}_w^{ \dag}}{\sigma_{\it w}^2}\leq\hat{\lambda}_w\,\mathbf{I}_{N_a}$, which implies that Willie's channel has limited capabilities such as a low receiver sensitivity or Willie cannot approach a certain area around the transmitter. Also, such a compound channel setting is a scenario where Alice is trying to hide the presence of the communication session in the existence of multiple adversaries. It is worth noting that the secrecy capacity of such a class of channels, i.e., the secrecy capacity of the compound channel, is the worst-case secrecy capacity \cite[Theorem 4]{schaefer2015secrecy}. Similarly, we show in Appendix \ref{App:compound} that the scaling of covert information of the compound MIMO AWGN channels equals the worst-case scaling over such class of channels using results in \cite{schaefer2015secrecy} and \cite{compound_results}. This means that there exists a code that achieves this scaling over the whole class. The following corollary gives the scaling of covert information of the compound MIMO AWGN channels.
\begin{cor}\label{compound_channel_1}
The scaling of the maximum number of covert nats that can be transmitted reliably over the compound MIMO AWGN channels with a known legitimate receiver's CSI is given as follows:
\begin{equation} 
L=\sum_{i=1}^{N}\frac{\sqrt{2\,c_i}\,\lambda_{\it b,i}}{\sigma_{\it b}^2\,\hat{\lambda}_w},
\end{equation}
while ensuring that the illegitimate receiver's sum of the detection error probabilities is lower bounded by $(1-\delta_{kl})$. Moreover, when the legitimate receiver's channel is well-conditioned, the scaling is given as:
\begin{equation} 
L=\frac{\sqrt{2\,N}\,{\lambda}_{\it b}}{\sigma_{\it b}^2\,\hat{\lambda}_w}.
\end{equation}
\end{cor}
\subsubsection{{\bf Unit-Rank MIMO AWGN Channels}} In this subsection, unit-rank MIMO channels are analyzed using the physical modeling of MIMO channels in \cite{tse2005fundamentals}. Without loss of generality, the focus is on evenly-spaced uniform linear antenna arrays. We consider the line-of-sight (LoS) MIMO channels and define $\mathbf{u}_a(\Omega_a)$ as the unit spatial transmit signature in the directional cosine, $\Omega_a\triangleq \cos\phi_a$, where $\phi_a$ is the angle of departure of the LoS from the transmit antenna array. Similarly, we define $\mathbf{u}_r(\Omega_r)$ as the unit spatial receive signature in the directional cosine, $\Omega_r\triangleq \cos\phi_r$, where $\phi_r$ is the angle of incidence of the LoS onto the receive antenna array where $r\in\left\{b,w\right\}$. 

Under the assumption of smaller array dimension than the distance between the transmitter and the receiver and equal LoS path attenuations, $\xi_r$, for all transmit-receive antenna pairs, the channel gain matrix is given by:
\begin{equation}
{\bf H}_r =\sqrt{\lambda_r}\,\exp{\left(-j2\pi d\right)}\,\mathbf{u}_r(\Omega_r)\,\mathbf{u}_a ^\dag(\Omega_a), 
\end{equation}
where $\lambda_r$ is a unique non-zero singular value such that $\lambda_r=\xi_r^2\,N_a\,N_r$, $N_a$ and $N_r$ are the number of transmit and receive antennas, respectively, $r\in\left\{b,w\right\}$ and $d$ is the distance between the first transmit and the first receive antennas normalized to the carrier wavelength. Consequently, the channel gain matrix is given by:
\begin{equation}
{\bf H}^\dag_r\,{\bf H}_r =\lambda_r\,\mathbf{u}_a(\Omega_a)\,\mathbf{u}_a ^\dag(\Omega_a).
\end{equation}

In addition, the angle, $\theta$, between any two different spatial transmit signatures is related to the directional transmit cosines as follows:
\begin{equation}
\abs{\cos\theta}=\abs{\mathbf{u}_a^\dag(\Omega_1)\,\mathbf{u}_a(\Omega_2)}=\abs{\frac{\sin(\pi\,L_a\,\Omega_a)}{N_a\sin(\pi\,\,L_a\,\Omega_a/N_a)}},
\end{equation}
where $\Omega_a\triangleq\Omega_2-\Omega_1$, $L_a\triangleq N_a\,\Delta_a$ is the length of the transmit antenna array normalized with respect to the carrier wavelength, and $\Delta_a$ is the normalized transmit antenna separation. Moreover, the quantity $\abs{\cos\theta}=\abs{f(\Omega_a)}$ is a periodic function with a period $N_a/L_a$ and nulls at $\Omega_a=k/L_a$, $k=1,\dots,N_a-1$. Fig. \ref{fig:pattern} depicts the periodic function $\abs{f(\Omega_2-\Omega_1)}$ for a fixed transmit direction, $\Omega_1$, different values of the number of transmit antennas, and the normalized array length. 
\begin{figure}[h]
 \centering
 \includegraphics[scale=0.3]{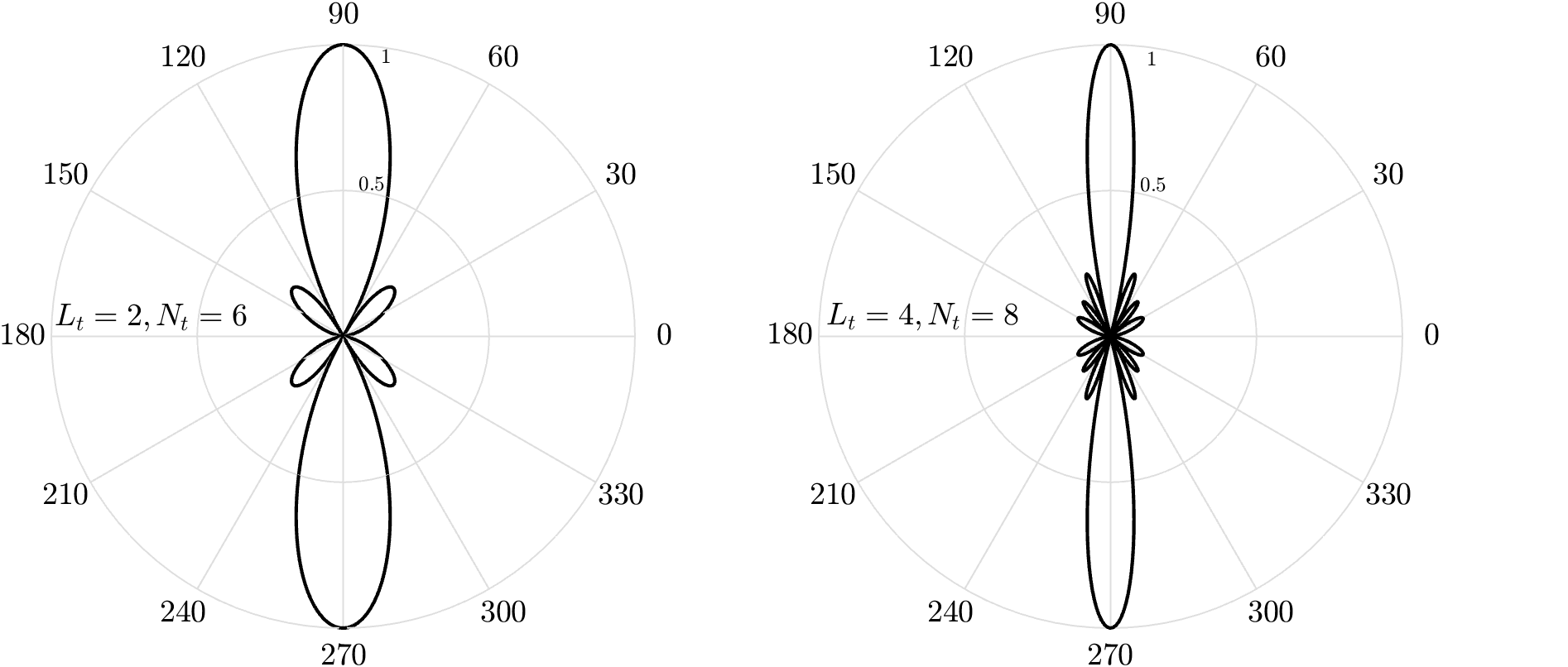}
 \caption{Transmit beam pattern aimed at $90^{\circ}$ for different values of the number of transmit antennas, $N_a$, and the normalized array length, $L_a$.}
 \label{fig:pattern}
\end{figure}

Alice tries to find the optimal spatial transmit signature in the directional cosine, $\Omega_a^*$, that maximizes the achievable covert coding rate. The optimal directional cosine, $\Omega_a^*$, can be characterized as follows:
\begin{equation}
\begin{aligned}\label{optimal_transmit}
\Omega_a^*=&\argmax_{-1\leq\Omega_a\leq 1}
& & \mathrm{\abs{f(\Omega_a-\Omega_b)}^2} \\
& \text{subject to:}
& &{\frac{P\,\tilde{\lambda}_{w}}{\sigma_w^2}}-\log{\left(\frac{P\,\tilde{\lambda}_{w}}{\sigma_w^2}+1\right)}\leq\frac{2\delta_{kl}^2}{n},
\end{aligned}
\end{equation}
where $\tilde{\lambda}_b={\lambda}_b\abs{f(\Omega_a-\Omega_b)}^2$ and $\tilde{\lambda}_w={\lambda}_w\abs{f(\Omega_a-\Omega_w)}^2$ are the eigenvalues of Bob's and Willie's channels, respectively, after projecting on the directional cosine, $\Omega_a$. Therefore, the following corollary gives the scaling for unit-rank MIMO AWGN channels.
\begin{cor}
The scaling of the maximum number of nats that can be transmitted covertly and reliably over unit-rank MIMO AWGN channels is given by:
\begin{equation}
L=\frac{\sqrt{2}\,{\sigma_w^2}\,\tilde{\lambda}_{\it b}}{\sigma_{\it b}^2\,\tilde{\lambda}_{w}}=\frac{\sqrt{2}\,{\sigma_w^2}\,{\xi}^2_{\it b}\,N_b\,\abs{f(\Omega_a^*-\Omega_b)}^2}{\sigma_{\it b}^2\,\xi^2_{w}\,N_w\abs{f(\Omega_a^*-\Omega_w)}^2},
\end{equation}
where $\xi_b$ and $\xi_w$ are the LoS path attenuations of the legitimate and the illegitimate receivers' channels, respectively, while ensuring that the illegitimate receiver's sum of the detection error probabilities is lower bounded by $(1-\delta_{kl})$.
\end{cor}
\begin{remark}
The scaling law of SISO AWGN channels, with isotropic transmit and receive antennas, is given by $L=\frac{\sqrt{2}\,{\sigma_w^2}\,{\xi}^2_{\it b}}{\sigma_{\it b}^2\,\xi^2_{w}}$, which coincides with the result in \cite{wang2016fundamental} when $\sigma_b=\sigma_w$, $\xi_{b}=\xi_{w}$ and the channel input is real.
\end{remark}

Alice can achieve a non-zero covert rate by transmitting in the spatial transmit signature in the null directional cosine of Willie's channel, i.e., $\Omega_n=\Omega_w+k/L_a$, $k=1,\dots,N_a-1$. In this case, the KL constraint is not active and $\delta_{kl}=0$. Hence, the gradient of the Lagrange function of the covert coding rate, given by \eqref{eq_transmit_signature} in Appendix \ref{App:proof_prop1}, can be rewritten as:
\begin{equation}
\left[q^*+\frac{\sigma_{\it b}^2}{\tilde{\lambda}_b}\right]^{-1}=\mu,\ \ \Omega_b\neq\Omega_w+\frac{k}{L_a}\mod{\frac{1}{\Delta}},
\end{equation}
where $\tilde{\lambda}_b={\lambda}_b\abs{f(\Omega_a)}^2$ and $\Omega_a=\Omega_n-\Omega_b=\Omega_w-\Omega_b+k/L_a$, $k=1,\dots,N_a-1$. The following corollary gives the achievable positive covert rate.
\begin{cor}\label{positive_rate_cor}
A positive covert rate can be achieved over unit-rank MIMO AWGN channels by transmitting in the spatial transmit signature in the null directional cosine of the illegitimate receiver's channel, i.e., $\Omega_n=\Omega_w+k/L_a$, $k=1,\dots,N_a-1$, namely, null steering. Hence, the following covert rate is achievable:
\begin{equation}
\begin{aligned}
R_{c}&=\lim_{n\rightarrow \infty}{R}_{a}(n,\epsilon,\delta_{kl})\\
&\ =\log\left(\frac{P{\lambda}_b\abs{f(\Omega_w-\Omega_b+k/L_a)}^2}{\sigma_{\it b}^2}+1\right),
\end{aligned}
\end{equation}
where ${R}_a(n,\epsilon,\delta_{kl})$ is an achievable non-diminishing covert coding rate and the choice of $k$ that maximizes $R_{c}$ is given by: 
\begin{equation}\label{choice_of_k}
k=\argmax_{1\leq i\leq N_a-1} {\abs{f(\Omega_w-\Omega_b+\frac{i}{L_a})}}.
\end{equation}
\end{cor}

Fig. \ref{fig:steering} shows an example of a transmit beam that its null is steered in the direction of Willie, while achieving a positive covert rate. Similarly, Alice can achieve a positive covert coding rate by utilizing the null-space precoding if a null-space between Alice and Willie exists. Thus, the loss in the capacity of unit-rank MIMO AWGN channels due to steering is the difference between the capacity and the positive covert rate achieved in Corollary \ref{positive_rate_cor}.
\begin{figure}[h]
 \centering\includegraphics[scale=0.45]{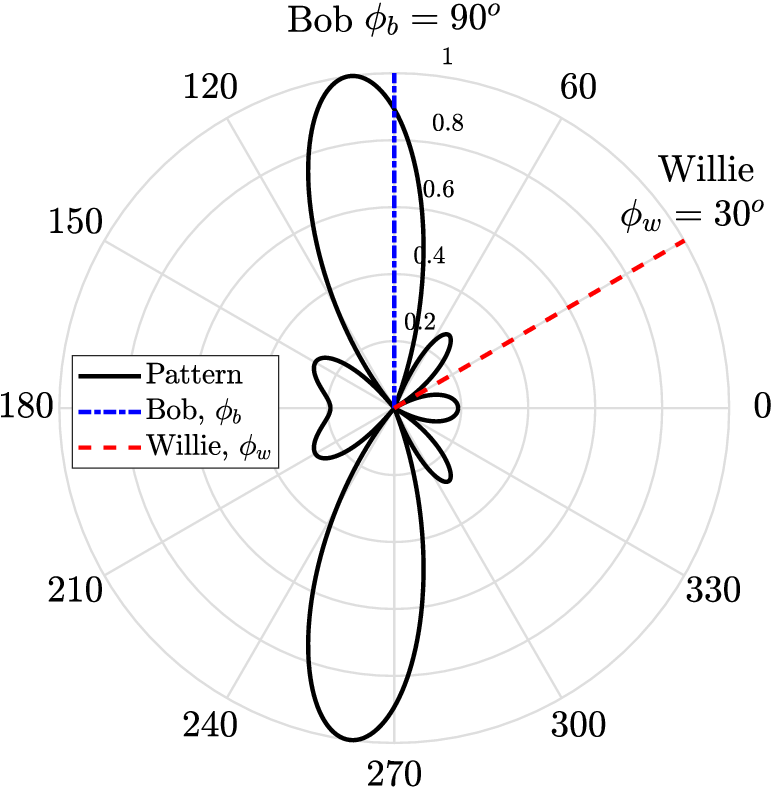}
 \caption{Steered transmit beam for $N_a=4$, and $L_a=2$.}
 \label{fig:steering}
\end{figure}

\section{Covert Communication without a Secret Codebook:\\ Joint Information Secrecy and Energy-Undetectability}\label{sec:covert and secure}
Wireless secrecy has been an active topic of research for more than a decade now \cite{Hassibi_secrecy_capacity,Liang_book}. One may be inclined to think that, energy-undetectable communication is sufficient to achieve IS automatically. Indeed, if Willie cannot detect the presence of a session in the first place, how can Willie decode the message? Furthermore, the capacity of energy-undetectable communication is zero in many cases, while existing results show that the IS can be achieved with a positive capacity. 

We will show that, surprisingly, this insight is not accurate. The major difference here is the assumption of the knowledge of the codebook. If one assumes the knowledge of the codebook at Willie, the optimal detector for the session between Alice and Bob is no longer the energy detector. Indeed, given the codebook, the transmitted signal belongs to a finite set of discrete possibilities. Thus, the optimal test is coherent detection, which reduces to a linear correlator under Gaussian noise, rather than energy detection. Therefore, an energy detector failing to identify the active session does not necessarily imply covertness. We will show that there are scenarios in which a session between Alice and Bob is energy-undetectable, yet IS condition is not satisfied, i.e., Willie can decode a few messages even at an arbitrarily low amount of received energy. Thus, to transmit a covert message, one should guarantee both confidentiality and energy-undetectability. Fig. \ref{Fig_secrecy_vs_covert} depicts the relation between communication with IS, covert communication, and energy-undetectable communication.
\begin{figure}[h]
 \centering
 \includegraphics[scale=.5]{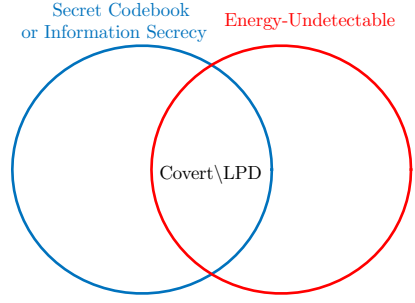}
 \caption{Covert communication, communication with IS, and energy-undetectable communication.\label{Fig_secrecy_vs_covert}}
\end{figure}

\begin{definition}\label{def_noisier}
The illegitimate receiver's channel is called noisier (less capable) than the legitimate receiver's channel in the $i^{th}$ eigen-direction, $1\leq i\leq N$, when
\begin{equation}
\frac{{\lambda}_{\it b,i}}{\sigma_{\it b}^2}>\frac{{\lambda}_{\it w,i}}{\sigma_{\it w}^2}.
\end{equation}
\end{definition}

In this section, the codebook is not kept secret but Willie's channel is noisier than Bob's channel as in \cite{Bloch_resolvability,constant_rate}. Although we do not impose such a constraint, noisier Willie's channel is required to achieve IS. To state the main problem of this section, we start with a modified version of Definition \ref{def}, which is stated as follows:
\begin{definition}\label{code_IS}
A $(2^{{\lceil nR\rceil}},n,\epsilon,\delta_{kl},\delta_{s})$-code consists of:
\begin{enumerate}
\item A uniformly distributed message set, $\mathcal{M}_n$.
\item A stochastic encoder, $f:\mathcal{M}_n\times\mathcal{T}_n\mapsto \mathbb{C}^{n}$, that assigns a codeword ${x}^n(m,t)$ to each message\footnote{The randomization variable $T$ is independent of the message $M$.} $m\in\mathcal{M}_n$ and $t\in\mathcal{T}_n$ under an average transmit power constraint $P$ on every codeword, ${\bf x}^n(m,t)\triangleq\left\{{\bf x}_1(m,t),\ \dots,{\bf x}_n(m,t)\right\}$, as follows: $\sum_{i=1}^n{\bf x}^{\dag}_i(m,t)\,{\bf x}_i(m,t)\leq nP$.
\item The decoder, $\phi$.
\item The encoder-decoder pair satisfies:\\
a) Average probability of decoding error constraint: $P_e^{(n)}(\mathcal{C}_n)\leq \epsilon$.\\
b) KL constraint: $\mathcal{D}(\mathbb{P}_{\mathbf{Y}_w^n} \parallel \mathbb{P}_{\mathbf{Z}_w^n}) \leq 2\,\delta_{kl}^2$. \\
c) IS constraint (one of the following):
\begin{equation}\label{eq:constraints}
\begin{array}{lrl}
\mbox{Weak secrecy:} &\mathbb{S}=\frac{1}{n}\mathrm{I}(M;{\bf Y}_w^n)&\leq \delta_{s},\\
\mbox{Strong secrecy:}&\mathbb{S}=\mathrm{I}(M;{\bf Y}_w^n)&\leq \delta_{s},\\
\mbox{Effective secrecy:}&\mathbb{S}=\mathcal{D}(\mathbb{P}_{M,\mathbf{Y}_w^n} \parallel \mathbb{P}_M\mathbb{P}_{\mathbf{Z}_w^n})&\leq \delta_{s},
\end{array}
\end{equation}
where $\mathbb{P}_M$ is the probability distribution of the message, $M$, while $\mathbb{P}_{M,\mathbf{Y}_w^n}$ is the joint probability distribution of the message and Willie's observation and $\mathcal{D}(\mathbb{P}_{M,\mathbf{Y}_w^n} \parallel \mathbb{P}_M\,\mathbb{P}_{\mathbf{Z}_w^n})=\mathrm{I}(M;{\bf Y}_w^n)+\mathcal{D}(\mathbb{P}_{\mathbf{Y}_w^n} \parallel \mathbb{P}_{\mathbf{Z}_w^n})$.
\end{enumerate} 
\end{definition}

The asymptotic condition, $\lim_{n\rightarrow \infty}\mathbb{S}= 0$, corresponds to different IS notions that exist in the literature. The distribution of the randomization variable $T$ is known to all parties and the realizations are kept secret from both Bob and Willie. However, with a secret codebook, as in Section \ref{sec:covert}, the realizations of $T$ are known to Bob and are kept secret from Willie. It should be clear that communication without a secret codebook but with IS is achieved at the expense of a lower secrecy rate than when a codebook is kept secret \cite{Liang_book}.
\begin{definition}
Maximal covert coding rate with IS:
\begin{equation}
{R}(n,\epsilon,\delta_{kl},\delta_{s})\triangleq\sup \left\{R:\,\exists\ \mbox{a}\ (2^{\lceil nR\rceil},n,\epsilon,\delta_{kl},\delta_{s})\mbox{-code}\right\}.
\end{equation}
\end{definition}
\begin{definition}
The scaling, with the blocklength, of the maximum number of nats that can be transmitted covertly with IS, while satisfying the average probability of decoding error constraint and the average power constraint, is defined as follows:
\begin{equation}
\label{eq:L_s}
L_{S}\triangleq \lim_{\epsilon \downarrow 0} \varliminf_{n \rightarrow \infty} \sqrt{\frac{n}{2\,\delta_{kl}^2}}{R}(n,\epsilon,\delta_{kl},\delta_{s}).
\end{equation}
\end{definition}
\begin{problem} (Without a secret codebook) Characterize the scaling, $L_S$, of the maximum number of covert nats with IS over MIMO AWGN channels as defined in \eqref{eq:L_s}.
\end{problem}

To investigate the scaling law without a secret codebook, the secrecy capacity should be considered first. The secrecy capacity is the maximum achievable rate such that $R=R_e$ \cite{Liang_book} where $R_e\leq\liminf_{n\rightarrow \infty}{\frac{1}{n}H(M|{\bf Y}^n_w)}$, denoted as the equivocation rate, is Willie’s uncertainty about the message, $M$, given the channels' outputs ${\bf Y}^n_w$, i.e., the secrecy level at Willie. Hence, the secrecy capacity is given by:
\begin{equation}
C_s=\max_{(R,R)\in \mathcal{C_E}}R,
\end{equation}
where $(R,R)$ is the rate–equivocation pair and $\mathcal{C_E}$ is the capacity-equivocation region. Moreover, the secrecy capacity of MIMO wiretap channels is defined as follows \cite{Hassibi_secrecy_capacity}:
\begin{equation}\label{secrecy_capacity}
C_s=\max_{\substack{\mathbf{Q}_n \succeq \mathbf{0} \\ \mathbf{tr}(\mathbf{Q}_n) \leq P}}\,\left[R_b({\bf Q}_n)-R_w({\bf Q}_n)\right]^+,
\end{equation}
where $R_u({\bf Q}_n)=\log\abs{\frac{1}{\sigma_{u}^2}\,{\bf H}_{u}\,{\bf Q}_n\,{\bf H}_{u} ^\dag+{\bf I}_{N_{u}}}$ and $u\in\left\{b,w\right\}$. Using a similar approach that is adopted in the previous section, the following corollary gives the scaling law of covert communication without a secret codebook but with IS over MIMO AWGN channels. We provide a sketch for the achievability proof in Appendix \ref{App:proof_corollary5}.
\begin{cor}
\label{thm:sec_covert}
The scaling of the maximum number of covert nats that can be transmitted reliably with IS over MIMO AWGN channels is given by:
\begin{equation}\label{optimization_secrecy}
\begin{aligned}
L_S= \varliminf_{n \rightarrow \infty} \sqrt{\frac{n}{2\,\delta_{kl}^2}}&\max_{\substack{\mathbf{Q}_n \succeq \mathbf{0} \\ 
\mathbf{tr}(\mathbf{Q}_n) \leq P}}
& & {\left[R_b({\bf Q}_n)-R_w({\bf Q}_n)\right]^+} \\
& \text{subject to:}
& &\mathcal{D}(\mathbb{P}_{\mathbf{Y}_w}\parallel\mathbb{P}_{\mathbf{Z}_w})\leq\frac{2\,\delta_{kl}^2}{n}.
\end{aligned}
\end{equation}

Moreover, the input distribution that maximizes the first-order approximation of the covert coding rate with IS, while minimizing $\mathcal{D}(\mathbb{P}_{\mathbf{Y}_w}\parallel \mathbb{P}_{\mathbf{Z}_w})$, is the zero-mean circularly symmetric complex Gaussian distribution with a covariance matrix $\mathbf{Q}_n$.
\end{cor}

\begin{remark}
The objective function in the optimization problem in the scaling with IS \eqref{optimization_secrecy} is not convex in general. Thus, we do not provide an explicit solution for this optimization problem. Instead, we derive our results in terms of the optimal power allocation in each eigen-direction as in the normalized KL divergence in each eigen-direction \eqref{optimal_power_implicit}. The optimal input covariance matrix that achieves the secrecy capacity in \eqref{secrecy_capacity} is investigated for special cases, for example, as in \cite{Optimum_Loyka}. 
\end{remark}
\begin{theorem}
\label{thm:secrecy_Scaling}
The scaling of the maximum number of covert nats that can be transmitted reliably with IS over MIMO AWGN channels is given by:
\begin{equation}
L_{S}=\sum_{i=1}^{N}\sqrt{2\,c_i}\left[\frac{\,{\sigma_w^2}\,{\lambda}_{\it b,i}}{\sigma_{\it b}^2\,{\lambda}_{w,i}}-1\right]^+,
\end{equation}
while ensuring that the illegitimate receiver's sum of the detection error probabilities is lower bounded by $(1-\delta_{kl})$.
\end{theorem}
\begin{IEEEproof}
The full proof can be found in Appendix \ref{App:proof_thm3}. We derive upper and lower bounds similar to the proof of Theorem \ref{thm:Scaling} by incorporating Corollary \ref{thm:sec_covert}. 
\end{IEEEproof}

Without the existence of a null-space, Theorem \ref{thm:secrecy_Scaling} coincides with Theorem 2 in \cite{constant_rate} for binary symmetric (BSC) channels where Alice can transmit $\mathcal{O}(\sqrt{n})$ covert nats that can be transmitted reliably with IS to Bob (without a secret codebook), namely, hidable and deniable nats. Also, it coincides with the results in \cite{Bloch_resolvability} for DMC and AWGN channels. Although we do not constrain the quality of any channel, interestingly, the expression in Theorem \ref{thm:secrecy_Scaling} is given in terms of the ratio between the quality of Bob's and Willie's channels, $\frac{{\lambda}_{\it u,i}}{\sigma_{\it u}^2},\ \forall i,\ u\in\left\{b,w\right\}$, which is inherited from the IS constraint. In each eigen-direction, this ratio should be greater than 1 to achieve a positive scaling in this direction. Thus, covert communication without a secret codebook but with IS cannot be achieved unless Bob's channel is less noisy than Willie's channel at least in one eigen-direction. In contrast, covert communication with a secret codebook can be achieved regardless of the quality of any channel but at a rate that depends on that ratio.

As long as the channel to Willie is noisier than the channel to Bob, Alice exploits the dominance of the channel quality of Bob over the channel quality of Willie in each eigen-direction. This allows using an ensemble of the public codebook and transmitting messages with IS to Bob and without being decoded by Willie \cite{Liang_book,constant_rate}. To ensure this, the stochastic encoder uses dummy messages and chooses the confidential message uniformly with a randomization rate that is determined by Willie's channel quality \cite{Poor_PLS}. Hence, Willie can be overwhelmed by the dummy messages and cannot decode but still has the ability to detect, i.e., the IS constraint is necessary but not sufficient to achieve covert communication. Therefore, the KL constraint (which is also not sufficient without the IS constraint) can be added to guarantee that Willie cannot detect the transmission as well as cannot decode. A closely related metric is the effective secrecy capacity that is developed in \cite{effectivesecrecy} where Willie tries to detect whether the transmission is meaningful or not, namely, stealth communication with IS. The effective secrecy capacity is positive and is similar to the weak and the strong secrecy capacities if there is a distribution $\mathbb{P}_{\mathbf{Z}_w^n}$ such that $\mathbb{P}_{\mathbf{Y}_w^n}=\mathbb{P}_{\mathbf{Z}_w^n}$.

\subsection{Scaling Laws of Specific MIMO AWGN Channels}
\begin{cor}
The scaling of the maximum number of covert nats that can be transmitted reliably with IS over well-conditioned MIMO AWGN channels is given by:
\begin{equation}
L_{S}=\sqrt{2\,N}\left[\frac{\,{\sigma_w^2}\,{\lambda}_{\it b}}{\sigma_{\it b}^2\,{\lambda}_{w}}-1\right]^+.
\end{equation}
\end{cor}
\begin{cor}\label{compound_channel_2}
The scaling of the maximum number of covert nats that can be transmitted reliably with IS over the compound MIMO AWGN channels with a known legitimate receiver's CSI is given as follows:
\begin{equation} 
L_S=\sum_{i=1}^{N}\sqrt{2c_i}\left[\frac{{\lambda}_{\it b,i}}{\sigma_{\it b}^2\hat{\lambda}_{w}}-1\right]^+.
\end{equation}
Moreover, when the legitimate receiver's channel is well-conditioned, the scaling is given as:
\begin{equation} 
L_S=\sqrt{2\,N}\left[\frac{{\lambda}_{\it b}}{\sigma_{\it b}^2\,\hat{\lambda}_{w}}-1\right]^+.
\end{equation}
\end{cor}
\begin{cor}
The scaling of the maximum number of covert nats that can be transmitted reliably with IS over unit-rank MIMO AWGN channels is given by:
\begin{equation}\begin{aligned}
L_{S}&=\sqrt{2}\left[\frac{\,{\sigma_w^2}\,\tilde{\lambda}_{\it b}}{\sigma_{\it b}^2\,\tilde{\lambda}_{w}}-1\right]^+\\
&=\sqrt{2}\left[\frac{{\sigma_w^2}\,{\xi}^2_{\it b}\,N_b\,\abs{f(\Omega_a^*-\Omega_b)}^2}{\sigma_{\it b}^2\,\xi^2_{w}\,N_w\,\abs{f(\Omega_a^*-\Omega_w)}^2}-1\right]^+.
\end{aligned}\end{equation}
\end{cor}

\subsection{Covert communication, communication with IS, and energy-undetectable communication}\label{discussion}

We conclude this part of the paper with a brief discussion about the relationship between communication with IS, covert communication, and energy-undetectable communication (under the KL constraint). Without the KL constraint and Bob knows the codebook, Bob utilizes an optimum decoder, which is a linear correlator, that minimizes the probability of decoding error. Hence, for a sufficiently large $n$, communication is reliable with an information rate $R=\lim_{n\rightarrow\infty}\frac{1}{n}\mathrm{I}(M;{\bf Y}^n_b)=\mathcal{O}(1)$. Meanwhile, with a secret codebook and the one-time pad is utilized, the information leakage to Willie (about a message $M$) is zero \cite{Shannon}, i.e., $\mathrm{I}(M;{\bf Y}^n_w)=0$. Also, without a secret codebook but IS is utilized, the information leakage is zero for a sufficiently large $n$, i.e., $\lim_{n\rightarrow\infty}\mathrm{I}(M;{\bf Y}^n_w)=0$. In both cases, Willie cannot decode reliably. However, with the knowledge of the codebook construction and the encoding criteria, Willie can perform statistical hypothesis testing to detect whether there is a transmission or not as given in \eqref{HST}. In this case, the optimum detector that minimizes the sum of detection error probabilities is an energy detector. Although Willie cannot decode reliably, the energy detector can detect the presence of communication.

To prevent communication from being energy detected and achieve covert communication, the transmitted power should be a decreasing function of the message length to satisfy the KL constraint. Thus, for a sufficiently large message length, the transmitted power goes to zero and $\alpha+\beta$ goes to one. Hence, the number of possible messages cannot be increased, while achieving a small probability of decoding error. This illustrates that the information-theoretic capacity of the energy-undetectable channel, which satisfies the KL constraint, is zero. Although the information rate becomes zero, i.e., $R=\lim_{n\rightarrow\infty}\frac{1}{n}\mathrm{I}(M;{\bf Y}^n_b)=0$, the probability of decoding error at Bob still decays exponentially with the codeword length \cite{bash2013limits} and reliable information $\mathcal{O}(\sqrt{n})$ nats can be transmitted covertly, i.e., $\mathrm{I}(M;{\bf Y}^n_b)=\mathcal{O}(\sqrt{n})$ nats.

Now, we further consider the case when Alice transmits energy-undetectable information, i.e., the KL constraint is satisfied, but communication is without either a secret codebook or IS. Although the information rate is zero due to the KL constraint, both Bob and Willie can decode and get information $\mathcal{O}(\sqrt{n})$ nats reliably, i.e., both $\mathrm{I}(M;{\bf Y}^n_b)$ and $\mathrm{I}(M;{\bf Y}^n_w)$ are $\mathcal{O}(\sqrt{n})$ nats. Thus, Willie is still able to detect the presence of communication by a linear correlator when communication is without both a secret codebook and IS even if the transmitted signal is energy-undetectable. This means that the KL constraint is not sufficient without either a secret codebook or IS. To achieve covert communication, Willie should not be able either to decode the message or to detect the signal energy. In other words, to transmit a covert message, one should guarantee both confidentiality and energy-undetectability. Table \ref{discussion_table} concludes this discussion\footnote{The notation $a(n)\rightarrow b$ is equivalent to $\lim_{n\rightarrow\infty}a(n)=b$.}. Also, Fig. \ref{Fig_secrecy_vs_covert} captures the main idea. 

\begin{table*}
\centering
\caption{Different cases for Covert communication, communication with IS and energy-undetectable communication.\label{discussion_table}}
\begin{tabular}{|c|c|c|c|c|c|}
 \hline
 {Cases} & {Codebook} & {Power} & {Communication with IS} & {Energy-undetectable Communication}& {Covert Comm.} \\ \hline
 \text{1} & {Public} & {Const.} & {No: $\mathrm{I}(M;{\bf Y}^n_w)=\mathcal{O}({n})$} & {No: $\alpha+\beta\rightarrow0$, $\frac{1}{n}\mathrm{I}(M;{\bf Y}^n_w)=\mathcal{O}({1})$}& {No} \\ \hline
 \text{2} & {Public} & {$\propto\frac{1}{\sqrt{n}}$} & {No: $\mathrm{I}(M;{\bf Y}^n_w)=\mathcal{O}(\sqrt{n})$} & {Yes: $\alpha+\beta\rightarrow1$, $\frac{1}{n}\mathrm{I}(M;{\bf Y}^n_w)\rightarrow0$}& {No} \\ \hline
 \text{3} & {Secret} & {Const.} & {Yes: $\mathrm{I}(M;{\bf Y}^n_w)= 0,\ {\forall n}$} & {No: $\alpha+\beta\rightarrow0$, $\frac{1}{n}\mathrm{I}(M;{\bf Y}^n_w)=0,\ {\forall n}$}& {No} \\ \hline
 \text{4} & {Secret} & {$\propto\frac{1}{\sqrt{n}}$} & {Yes: $\mathrm{I}(M;{\bf Y}^n_w)=0,\ {\forall n}$} & {Yes: $\alpha+\beta\rightarrow1$, $\frac{1}{n}\mathrm{I}(M;{\bf Y}^n_w)=0,\ {\forall n}$}& {Yes} \\ \hline
 \text{5} & {Public with IS} & {$\propto\frac{1}{\sqrt{n}}$} & {Yes: $\mathrm{I}(M;{\bf Y}^n_w)\rightarrow 0$} & {Yes: $\alpha+\beta\rightarrow1$, $\frac{1}{n}\mathrm{I}(M;{\bf Y}^n_w)\rightarrow0$}& {Yes} \\ \hline
 \end{tabular}
\end{table*}


\subsection{Using the product distribution in the achievability proofs}\label{using the product distribution} With a secret codebook, according to Definition \ref{secret codebook}, the probability of each codeword observed at Willie is given by a product distribution\footnote{The product distribution is the product of the marginal distributions.}. The reason is that the received codeword seems to Willie as if it is just sampled at the moment of the transmission. This scenario is adopted in the achievability proofs in \cite{bash2013limits}. 

In general, under an average cost constraint, codewords that do not satisfy the constraint are discarded. This violates the i.i.d. assumption such that we cannot use the product distribution for Bob in the achievability proofs. Similarly, when the codebook is public, the product distribution cannot be used to derive the single-letter KL divergence at Willie. However, for a large enough blocklength, randomly-generated codewords satisfy the constraint by the typical average lemma. Satisfying the input cost constraint prevents discarding any codeword, and thus, the i.i.d. assumption is still valid. Equivalently, instead of discarding the violating codewords, we let the decoder declare an error, the probability of which vanishes for a sufficiently large blocklength. As a result, although the codebook is public, using the product distribution for both Bob and Willie is a plausible assumption. Also, when the codebook is public, we use the IS constraint to prevent Willie from decoding, and thus, the detection problem reduces to energy detection. Therefore, the product distribution can be safely used.

In an identical scenario, the product distribution is used for the average probability of decoding error analysis as in \cite[Theorem 9.1.1]{cover2012elements}, \cite[Theorem 3.2]{elgamal_kim_2011}, \cite[Theorem 1.1]{bash2013limits}, and \cite{wang2016fundamental}. Even if the channel suffers from a fading process with memory, the memoryless channel and the product distribution assumptions are still valid when the CSI is known to Alice and Bob \cite[Ch 2]{mimo_channel_goldsmith}. Also, under a peak power constraint, the i.i.d. assumption is addressed successfully in the achievability proof of \cite[Theorem 1.2]{bash2013limits} by using a sub-optimal decoder at Bob that considers each observation independently. For more involved achievability proofs that exploit a modified version of typicality, we refer to \cite[Remark 3, eq. 76]{Bloch_resolvability} and \cite[eq. 18]{constant_rate}. Therein, random coding satisfies both the probability of decoding error and the energy-undetectability constraints with probability $1-\theta_n$, where $\theta_n$ decays exponentially and even super exponentially in some scenarios.
\section{Can We Overcome the square-Root Law?}\label{sec:asymptotic}

This section answers an interesting question and gives important insight into the asymptotic behavior of covert communication with unknown CSI of Willie's channel. In this regime, Alice transmits directly to Bob without steering the null of its transmit beam to the direction of Willie's channel. Hence, the covert capacity with unknown CSI of Willie's channel is zero. However, increasing the number of transmit antennas overcomes the square-root law of covert communication, and hence, a positive covert capacity is achieved. As the number of transmit antennas goes to infinity, namely, the massive MIMO limit \cite{amr}, the maximal covert coding rate converges to the maximal coding rate of MIMO AWGN channels. The same concept applies to communication with IS but the scope of this section is limited to covert communication.

In the light of Theorem \ref{thm:Kn_covert}, this section considers the same problem: covert communication with a secret codebook. Thus, in this section, we study the effect of increasing the number of transmit antennas for the same optimization problem, while Theorem \ref{thm:Kn_covert} considers an arbitrary number of transmit antennas. Specifically, we consider the first-order approximation of the maximum covert coding rate but under a different asymptotic regime.

\subsection{Unit-Rank MIMO Channels}
With unknown CSI of Willie's channel, Alice transmits in the spatial transmit signature in the directional cosine of Bob's channel. To give more insight utilizing the antenna array design, we consider the received pattern at Willie, which is given by $\abs{f(\Omega)}=\abs{f(\Omega_b-\Omega_w)}$, in the following two cases:
\begin{itemize}
 \item Case (1): Fixed normalized array length, $L_a$, of the transmit antenna array:
 \begin{equation}\label{case1}
 \lim_{N_a\rightarrow\infty} \abs{\frac{\sin(\pi\,L_a\,\Omega)}{N_a\,\sin(\pi\,L_a\,\Omega/N_a)}}=\frac{\sin(\pi\,L_a\,\Omega)}{\pi\,L_a\,\Omega}. 
 \end{equation}
This means that as the number of transmit antennas increases, the main lobe does not change and all other lobes decreases, which increases the maximum covert coding rate as long as Willie is not aligned to the main lobes that are centered at $\phi_b$ and $2\pi-\phi_b$, respectively, with a beamwidth equals to $2/L_a$.
\item Case (2): Fixed normalized antenna separation, $\Delta_a$:
\begin{equation}\label{case2}
 \lim_{N_a\rightarrow\infty} \abs{\frac{\sin(\pi\,N_a\,\Delta_a\,\Omega)}{N_a\,\sin(\pi\,\Delta_a\,\Omega)}}=0.
 \end{equation}
In this case, the width of the main lobe decreases and becomes very directive, pencil beam, as well as all other lobes decrease substantially. Thus, Willie can not receive the transmitted signal wherever Willie is not aligned to the spatial transmit direction of Bob, i.e., $\Omega_b\neq\Omega_w+\frac{k}{L_a}\mod{\frac{1}{\Delta_a}},\ k=1,\dots,N_a-1$. Therefore, a positive covert capacity can be achieved with zero probability of detection due to the high beamforming capability.
\end{itemize}
\begin{theorem}\label{prop:one}
Without the knowledge of the CSI of the illegitimate receiver and for a finite blocklength, the KL constraint of unit-rank MIMO AWGN channels is satisfied $\forall \delta_{kl}\geq 0$ as $N_a\rightarrow \infty$, under the following condition:
\begin{enumerate}
 \item Either a fixed normalized antenna separation, while the illegitimate receiver is not aligned to the spatial transmit direction of the legitimate receiver, i.e., $\Omega_b\neq\Omega_w+\frac{k}{L_a}\mod{\frac{1}{\Delta_a}},\ k=1,\dots,N_a-1$.
 \item Or a fixed normalized array length, while the illegitimate receiver is not aligned to the main lobes, which are centered at $\phi_b$ and $2\pi-\phi_b$, respectively, with a beamwidth equals to $2/L_a$.
\end{enumerate}

Moreover, a lower bound on the number of transmit antennas that achieves a predefined target probability of detection, $\delta_{kl}$, is given by: 
\begin{equation}
N_a\geq \frac{P\xi_w^2N_w\abs{\sin(\pi L_a\Omega)}^2}{\sigma_w^2\abs{\sin(\pi\,\Delta_a\,\Omega)}^2}\left[-W_{-1}\left(-e^{-\frac{2\delta_{kl}^2}{n}-1}\right)-1\right]^{-1},
\end{equation}
where $W_{-1}$ is a branch of the Lambert function\footnote{Bounds on and applications of the $W_{-1}$ branch of the Lambert function can be found in \cite{Lambert}.}. 
\end{theorem}
\begin{IEEEproof} The full proof can be found in Appendix \ref{App:proof_prop1}. We give a brief sketch as follows. With unknown Willie's CSI, Alice transmits in the spatial transmit signature in the directional cosine of Bob's channel. It suffices only to investigate the KL constraint for a sufficiently large number of transmit antennas. Utilizing the antenna array design, we show that the KL constraint is satisfied as $N_a$ goes to infinity, i.e., Willie does not receive the transmitted signal at all. Further, to estimate the number of transmit antennas that achieves a predefined target probability of detection, $\delta_{kl}$, we solve the KL constraint using the maximum transmit power.
\end{IEEEproof}

This result suggests that even without the knowledge of Willie's CSI, the maximal covert coding rate of unit-rank MIMO AWGN channels is asymptotically (with the number of transmit antennas) not diminishing and converges to the maximal coding rate of unit-rank MIMO AWGN channels. Hence, Alice can transmit $\mathcal{O}(n)$ nats covertly and reliably under some conditions. Thus, for a sufficiently large number of transmit antennas, Alice can achieve the capacity of unit-rank MIMO AWGN channels without being detected.

\begin{cor}
Under the stated conditions and for a sufficiently large number of transmit antennas, the covert capacity of unit-rank MIMO AWGN channels converges to the capacity of unit-rank MIMO AWGN channels, while the KL constraint is satisfied, i.e., $\lim_{n\rightarrow \infty}\lim_{N_a\rightarrow \infty}{R}(n,\epsilon,\delta_{kl})=\lim_{n\rightarrow \infty}\lim_{N_a\rightarrow \infty}{R}(n,\epsilon)$, for any given $\delta_{kl}\geq 0$. 
\end{cor}
\subsection{Multi-Path MIMO Channels}
Consider there are multiple reflected paths in addition to a LoS path where the $i^{th}$ path has an attenuation $\xi_i$ and makes an angle $\phi_{t,i}$ with the transmit antenna array and an angle $\phi_{r,i}$ with the receive antenna array, $\forall i$ and $r\in\left\{b,w\right\}$. The channel matrix, $\mathbf{H}$, is given by \cite{tse2005fundamentals} as follows:
\begin{equation}
\mathbf{H}=\sum_{i=1}\xi_i^p\sqrt{N_r\,N_t}\exp\left(-j\,2\,\pi\, d_i^p\right) \mathbf{u}_{r,i}^p(\Omega^p_{r,i})\,\mathbf{u}^{p\dag}_{t,i}(\Omega^p_{t,i}), 
\end{equation} 
where $\xi_i^p$ is the attenuation of the $i^{th}$ path, $d_i^p$ is the distance between the first transmit antenna and the first receive antenna along the $i^{th}$ path normalized to the carrier wavelength, $\mathbf{u}_{r,i}^p(\Omega_{r,i}^p)$ and $\mathbf{u}^{p}_{t,i}(\Omega_{t,i}^p)$ are the unit spatial receive and transmit signatures in the directional receive and transmit cosines of the $i^{th}$ path, $\forall i$, respectively. In addition, the channel matrix is full-rank if there exist at least $N$ paths such that $\Omega_{t,i}^p\neq\Omega_{t,j}^p\mod{\frac{1}{\Delta_a}}\mbox{ and }\Omega_{r,i}^p\neq\Omega_{r,j}^p\mod{\frac{1}{\Delta_r}}\ \forall i,j\ \mbox{and}\ i\neq j$, where $r\in\left\{b,w\right\}$.

Moreover, the channel matrix is well-conditioned if there exist at least $N$ paths such that the angular separation at the transmit array, $\abs{\Omega_t^p}=\abs{\Omega_{t,i}^p-\Omega_{t,j}^p},\ \forall i,j\ \mbox{and}\ i\neq j$, and at the receive array, $\abs{\Omega_r^p}=\abs{\Omega_{r,i}^p-\Omega_{r,j}^p},\ \forall i,j\ \mbox{and}\ i\neq j$, of each two paths are no less than $1/L_t$ and $1/L_r$, respectively, otherwise, paths are not resolvable.

The individual physical paths can be aggregated to the resolvable paths similar to the resolvable channel taps when modeling the multi-path fading channel \cite{tse2005fundamentals}. Hence, the $(i,j)^{th}$ channel gain in the angular domain consists of all paths whose transmit and receive directional cosines are within an angular window of width $1/L_t$ and $1/L_r$ around $l/L_t$ and $k/L_r$, respectively. Therefore, the orthonormal basis of the transmitted signal space, $\mathbb{C}^{N_t}$, and the received signal space, $\mathbb{C}^{N_r}$, are given by:
\begin{equation}
\mathcal{S}_u\triangleq\left\{\mathbf{u}_{u}\left(0\right), \mathbf{u}_{u}\left(\frac{1}{L_u}\right),\dots,\mathbf{u}_{u}\left(\frac{N_u-1}{L_u}\right)\right\},\end{equation}
where $u\in\left\{a,b,w\right\}$.

Let $\mathbf{U}_{t}$ and $\mathbf{U}_{r}$ be unitary matrices whose columns are the orthonormal vectors in $\mathcal{S}_a$ and $\mathcal{S}_r$, respectively. Then, the angular domain representation of the channel matrix is given by:
\begin{equation}
\begin{aligned}
\mathbf{H}^{g} &\triangleq \mathbf{U}_{r}^\dag\mathbf{H}\mathbf{U}_{t}\\
&=\sum_{i=1}\xi_i^p\sqrt{N_r\,N_t}\exp\left(-j\,2\,\pi d_i^p\right) \mathbf{u}_{r,i}^g(\Omega^g_{r,i})\mathbf{u}^{g\dag}_{t,i}(\Omega^g_{t,i})\\ 
&\stackrel{(a)}{=}\sum_{i=1}^{N_t}\xi_i\sqrt{N_r\,N_t}\exp\left(-j\,2\pi d_i\right) \mathbf{u}_{r,i}(\Omega_{r,i})\mathbf{u}^{\dag}_{t,i}(\Omega_{t,i}), \end{aligned} 
\end{equation}
where $\mathbf{u}_{t,i}^g(\Omega^g_{t,i})=\mathbf{U}_{t}^\dag\,\mathbf{u}^p_{t,i}(\Omega_{t,i}^p)$ and $\mathbf{u}^g_{r,i}(\Omega^g_{r,i})=\mathbf{U}_{r}^\dag\,\mathbf{u}^p_{r,i}(\Omega^p_{r,i})$, $ \forall i$, are the angular domain representation of the unit spatial transmit signature, $\mathbf{u}^p_{t,i}(\Omega^p_{t,i})$, in the directional cosine, $\Omega^p_{t,i}$ and the unit spatial receive signature, $\mathbf{u}^p_{r,i}(\Omega^p_{r,i})$, in the directional cosine, $\Omega^p_{r,i}$, of the $i^{th}$ path, respectively. In addition, $\mathbf{u}_{r,i}(\Omega_{r,i})$ and $\mathbf{u}_{t,i}(\Omega_{t,i})$ are the orthonormal vectors in $\mathcal{S}_t$ and $\mathcal{S}_r$, respectively, where $\xi_i$ is the attenuation of the $i^{th}$ angular window, $d_i$ is the distance between the first transmit antenna and the first receive antenna along the $i^{th}$ angular window normalized to the carrier wavelength, and (a) follows by aggregating all paths, i.e., unresolvable paths in the angular domain, that contribute in the same basis vector. For the $i^{th}$ path to contribute in the $l^{th}$ transmit basis vector and in the $k^{th}$ receive basis vector if:
\begin{equation}
\abs{\Omega^g_{u,i}-\frac{l}{L_u}}<\frac{1}{L_u},\ \forall i\ \mbox{and}\ l\in\left\{0,\dots,N_u-1\right\},
\end{equation}
where $u\in\left\{a,b,w\right\}$. Therefore, the angular representations of Bob's and Willie's channels power gain in terms of the resolvable paths are given by:
\begin{equation}
\mathbf{H}_u^{g\dag}\,\mathbf{H}_u^{g}=\sum_{i=1}^{N_a}\lambda_{u,i}\,\mathbf{u}_{u,i}(\Omega_{u,i})\,\mathbf{u}^{\dag}_{u,i}(\Omega_{u,i}), 
\end{equation}
where $u\in\left\{b,w\right\}$ and $\lambda_{u,i}=\xi^2_{u,i}\,N_u\,N_a$. 
\begin{theorem}\label{prop:two}
Without the knowledge of the CSI of the illegitimate receiver and for a finite blocklength, the KL constraint of multi-path MIMO AWGN channels is satisfied $ \forall \delta_{kl}\geq 0$ as $N_a\rightarrow \infty$, under the following condition:
\begin{enumerate}
 \item Either a fixed normalized antenna separation, while the illegitimate receiver is not aligned to any of the spatial transmit directions of the legitimate receiver, i.e., $\Omega_{b,i}\neq\Omega_{w,j}+\frac{k}{L_a}\mod{\frac{1}{\Delta_a}},\ k=1,\dots,N_a-1,\ \forall i\ \mbox{and}\ j$.
 \item Or a fixed normalized array length, while the illegitimate receiver is not aligned to the lobes that are centered at $\phi_{b,i}$ and $2\pi-\phi_{b,i},\ \forall i$, respectively, with a beamwidth equals to $2/L_a$.
\end{enumerate}
\end{theorem}
\begin{IEEEproof}The full proof can be found in Appendix \ref{App:proof_prop2}. We provide a brief sketch as follows. We rewrite the first-order approximation of the maximum covert coding rate using the angular domain representation of Bob's and Willie's channels. Then, we investigate the received pattern at Willie for a sufficiently large number of transmit antennas.
\end{IEEEproof}

It should be clear that the stated conditions in the previous theorem can not be easily satisfied. If so, Alice can asymptotically achieve the capacity of multi-path MIMO AWGN channels without being detected.

\section{Numerical Illustration}\label{numerical}

In this section, we provide some numerical results to illustrate the behavior of covert communication when either the number of transmit antennas or the blocklength scales up. Using the first-order approximation of the maximum coding rate of unit-rank MIMO AWGN channels, the number of achievable covert nats is given by:
\begin{equation}
n\,B\log\left(1+\frac{q\,\xi_b^2N_aN_b}{B\,\sigma_b^2}\right),\end{equation}
where $B$ is the occupied bandwidth, $q=\min\left\{P_{kl},P\right\}$ and \begin{equation}
P_{kl}=\frac{\sigma_w^2}{\xi_w^2N_a\,N_w|{f(\Omega)}|^2}\left[-W_{-1}\left(-e^{-\frac{2\,\delta_{kl}^2}{n}-1}\right)-1\right],
\end{equation}
while the achievable non-covert nats is given by: 
\begin{equation}
nB\log\left(1+\frac{P\,\xi_b^2N_a\,N_b}{B\,\sigma_b^2}\right).
\end{equation}

The simulation parameters are given as follows. The distance between Alice-Bob and Alice-Willie is $d=1$ km. We picked parameters that are typical in a millimeter-wave system, for which massive MIMO is a highly relevant regime. The LoS attenuation is $\xi_b^2=\xi_w^2=L\,d^{-\upsilon}$ where $\upsilon=2$ is the free-space path-loss exponent and $L=3.3\times 10^{-3}$ is the path-loss constant at frequency $73$ GHz. The noise power densities for both Bob and Willie are assumed to be equal and are given by $\sigma_b^2=\sigma_w^2=-174$ dBm/Hz. The maximum transmit power is $P=10$ dBm, the occupied bandwidth is $B=5$ MHz, $\Omega_w=\pi/4$, $\Omega_b=\pi/2$ and the probability of detection is $\delta_{kl}=10^{-2}$.

\begin{figure}[h]
 \centering\includegraphics[scale=0.41]{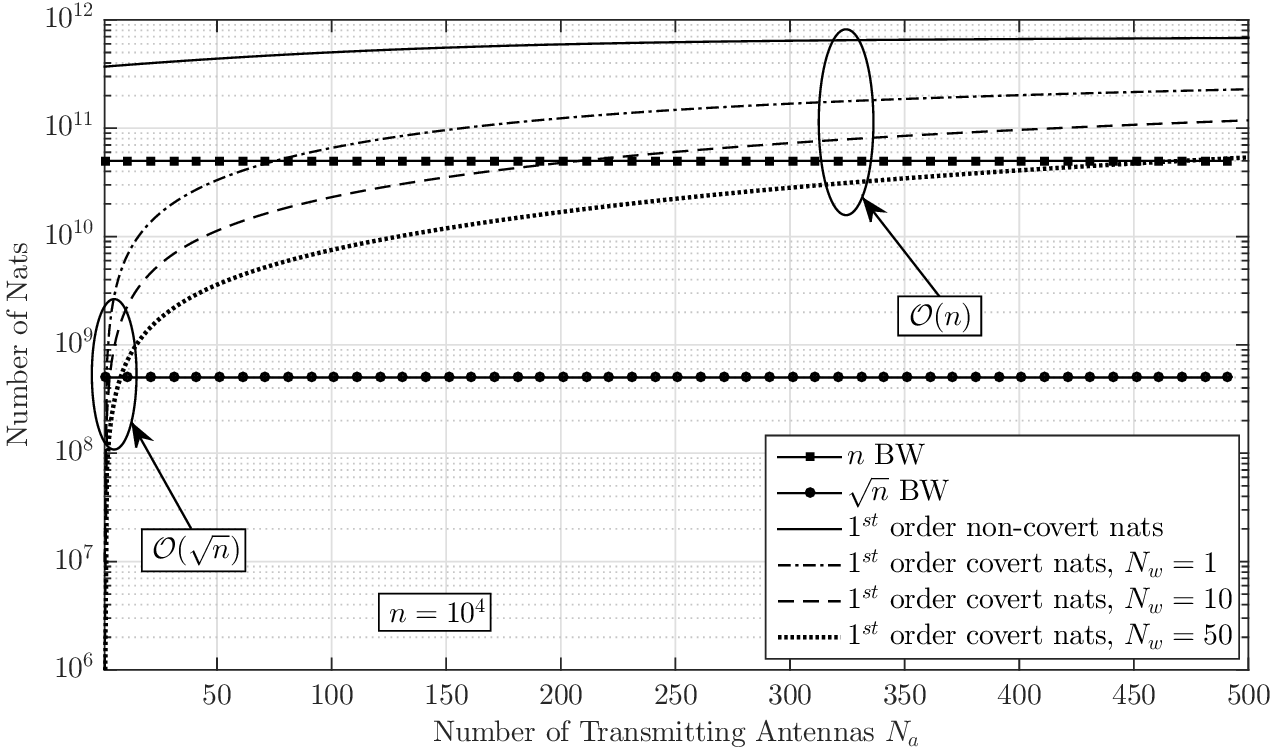}
 \caption{First-order approximation of maximum covert nats and non-covert nats vs. the number of transmit antennas.}
 \label{fig:nats_Na}
\end{figure}
\begin{figure}[h]
 \centering\includegraphics[scale=0.41]{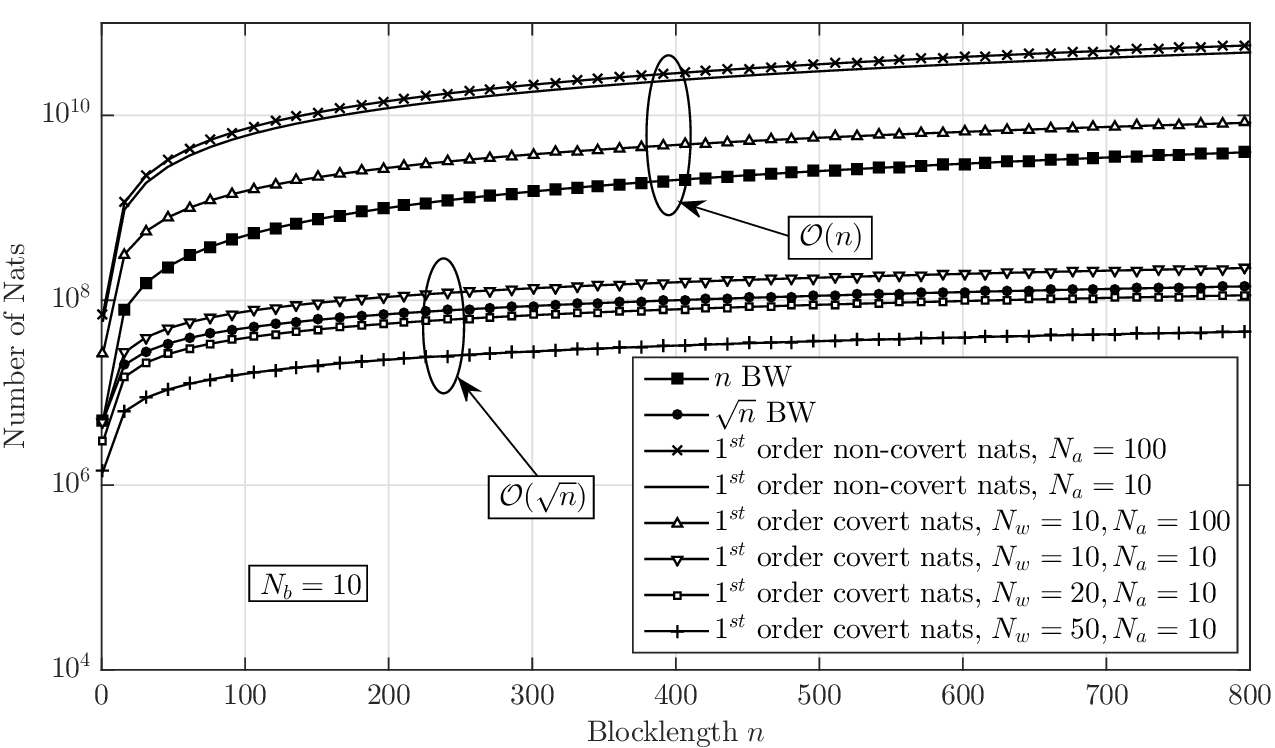}
 \caption{First-order approximation of maximum covert nats and non-covert nats vs. number of channel uses.}
 \label{fig:nats_n}
\end{figure}

Fig. \ref{fig:nats_Na} and \ref{fig:nats_n} compare the maximum number of achievable covert and non-covert nats that can be transmitted over unit-rank MIMO AWGN channels using the first-order approximation. Clearly, in Fig. \ref{fig:nats_Na}, the number of achievable covert nats converges to the number of achievable non-covert nats, for a blocklength, $n=10^4$, as the number of transmit antennas increases and for a different number of Willie's antennas, $N_w=1,10,50$. The number of achievable covert nats is $\mathcal{O}(n)$ instead of $\mathcal{O}(\sqrt{n})$. Although the number of antennas of Willie is increased from $1$ to $50$, the number of achievable covert nats is still $\mathcal{O}(n)$ for a large number of transmit antennas. In the case of a small number of transmit antennas, $N_a<10$, The number of achievable covert nats is $\mathcal{O}(\sqrt{n})$.

For a small number of transmit antennas, $N_a=10$, as in Fig. \ref{fig:nats_n}, the number of achievable covert nats is $\mathcal{O}(\sqrt{n})$ for a different number of Willie's antennas, $N_w=10,20,50$. In contrary, the number of achievable covert nats is $\mathcal{O}({n})$ for a fixed large number of transmit antennas, $N_a=100$.

\section{Conclusion}\label{sec:conclusion}
We studied the basic limits of covert communication over MIMO AWGN channels. One of the major findings is that, it is possible to achieve a positive covert capacity if the number of transmit antennas scales up sufficiently fast. We also showed that, in the massive MIMO limit, the covert capacity of MIMO AWGN channels is identical to its non-covert counterpart, as long as the illegitimate receiver is not aligned to any spatial transmit direction of the legitimate receiver. For an arbitrary number of transmit antennas, the achievable scheme involves the utilization of the null-space between the transmitter and the illegitimate receiver. 

On the other hand, the covert capacity is zero if there is no null-space between the transmitter and the illegitimate receiver. In this case, the maximum number of covert bits scales as $\mathcal{O}(\sqrt{n})$ with or without a secret codebook. Similar to communication with IS, covert communication without a secret codebook can be achieved only if the illegitimate receiver’s channel is noisier than the legitimate receiver’s channel. 

Future research could be directed towards investigating covert communication and communication with IS for delay-constrained applications, under different attack models, and exploiting artificial noise transmission to limit the illegitimate receiver’s detection capabilities.

\appendices
\section{Proof of Theorem \ref{thm:Kn_covert}}\label{App:proof_thm1}

\begin{IEEEproof}{\bf Converse:} 
We show that if there exists a sequence of $(2^{{\lceil nR\rceil}},n,\epsilon,\delta_{kl})$-codes with a vanishing average probability of decoding error, the scaling must be no greater than the right side of \eqref{eq:thm1}. Consider a given $(2^{\lceil nR\rceil},n,\epsilon,\delta_{kl})$-code\footnote{For a given codebook, ${X}_1,{X}_2,\dots,{X}_n$ are correlated, and hence, each of the observations ${Y}_{b,1},{Y}_{b,2},\dots,{Y}_{b,n}$ are correlated \cite{Wang_noncausal,Bloch_resolvability} and similarly ${Y}_{w,1},{Y}_{w,2},\dots,{Y}_{w,n}$.} that satisfies the KL constraint, $\mathcal{D}(\mathbb{P}_{\mathbf{Y}_w^n}\parallel \mathbb{P}_{\mathbf{Z}_w^n})\leq 2\,\delta_{kl}^2$, and $\lim_{n \rightarrow \infty} P^{(n)}_e=0$. 
\newline
{\bf Analysis:} 

1) {\em { Probability of decoding error analysis:}} Starting with Fano's and data processing inequalities: Let $R(n,\epsilon,\delta_{kl})=R$.
\begin{equation}\label{eq:converse1_prt1}
\begin{aligned}
nR=&\ H(M)=\ \mathrm{I}(M;{\bf Y}^n_b)+H(M|{\bf Y}^n_b)\\
\leq&\ \mathrm{I}(M;{\bf Y}^n_b)+n\,\epsilon_n
=\ h({\bf Y}^n_b)-h({\bf Y}^n_b|M)+n\,\epsilon_n\\
\overset{(a)}{\leq}&\ h({\bf Y}^n_b)-h({\bf Y}^n_b|M,{\bf X}^n)+n\,\epsilon_n\\
\overset{(b)}{=}&\ h({\bf Y}^n_b)-h({\bf Y}^n_b-{\bf H}_b\,{\bf X}^n|M,{\bf X}^n)+n\,\epsilon_n\\
=&\ h({\bf Y}^n_b)-h({\bf Z}^n_b|M,{\bf X}^n)+n\,\epsilon_n\\
\overset{(c)}{=}&\ h({\bf Y}^n_b)-h({\bf Z}^n_b)+n\,\epsilon_n\\
=&\ h({\bf Y}^n_b)-n\,\log\abs{\pi\,{\bf\Sigma}_b}+n\,\epsilon_n\\
\overset{(d)}{\leq}&\ \sum_{i=1}^nh({\bf Y}_{b,i})-n\,\log\abs{\pi\,{\bf\Sigma}_b}+n\,\epsilon_n\\
\overset{(e)}{\leq}&\ \sum_{i=1}^n\,\log\abs{\pi\,({\bf H}_b\,{\bf Q}_n\,{\bf H}_b ^\dag+{\bf\Sigma}_b)}-n\,\log\abs{\pi\,{\bf\Sigma}_b}+n\,\epsilon_n\\
=&\ n\,\log\abs{\pi\,({\bf H}_b\,{\bf Q}_n\,{\bf H}_b ^\dag+{\bf\Sigma}_b)}-n\,\log\abs{\pi\,{\bf\Sigma}_b}+n\,\epsilon_n\\
=&\ n\,\log\abs{\frac{1}{\sigma_b^2}\,{\bf H}_b\,{\bf Q}_n\,{\bf H}_b ^\dag+{\bf I}_{N_b}}+n\,\epsilon_n\\
\overset{(f)}{=}&\ n\,C_{c}({\bf Q}_n)+n\,\epsilon_n,
\end{aligned}\end{equation}
where $n\,\epsilon_n=1+P^{(n)}_e\,nR(n,\epsilon,\delta_{kl})$, $\epsilon_n$ tends to zero as $n$ goes to infinity by the assumption that $\lim_{n \rightarrow \infty} P^{(n)}_e=0$, $\mathbf{Q}_n$ is a decreasing function of $n$, ${\bf H}$ is stationary over $n$ channel uses, (a) follows since conditioning does not increase the entropy, (b) since translation does not change the entropy, (c) since the noise is independent of both the message and the transmitted codeword, (d) follows from the chain rule of entropy and removing conditioning, (e) follows since the maximum differential entropy of a continuous random vector is attained when the random vector has a zero-mean circularly symmetric complex Gaussian distribution with a covariance matrix, $\mathbf{Q}_n$, that attains the maximum, and (f) by the definition of the Gaussian vector channel capacity, while the optimal input covariance matrix, $\mathbf{Q}_n$, is chosen such that the KL constraint is satisfied.

2) {\em {KL constraint analysis:}} Following a similar approach of \cite{Hou_dependentRV,wang2016fundamental}:
\begin{equation}
\begin{aligned}\label{eq:converse1_prt2}
2\,\delta_{kl}^2\geq&\ \mathcal{D}(\mathbb{P}_{\mathbf{Y}_w^n}\parallel \mathbb{P}_{\mathbf{Z}_w^n})\\
=&\ -h({\bf Y}^n_w)+\mathbb{E}_{\mathbb{P}_{\mathbf{Y}_w^n}}\left[\log\frac{1}{f_{\mathbf{Z}_w^n}({\bf Z}^n_w)}\right]\\
=&\ -h({\bf Y}^n_w)+\sum_{i=1}^n\mathbb{E}_{\mathbb{P}_{\mathbf{Y}_{w,i}}}\left[\log\frac{1}{f_{\mathbf{Z}_w}({\bf Z}_{w,i})}\right]\\
\overset{(a)}{\geq}&\ \sum_{i=1}^n-h({\bf Y}_{w,i})+\sum_{i=1}^n\mathbb{E}_{\mathbb{P}_{\mathbf{Y}_{w,i}}}\left[\log\frac{1}{f_{\mathbf{Z}_w}({\bf Z}_{w,i})}\right]\\
=&\ \sum_{i=1}^n\,\mathcal{D}(\mathbb{P}_{\mathbf{Y}_{w,i}}\parallel \mathbb{P}_{\mathbf{Z}_{w,i}})\\
\overset{(b)}{\geq}&\ n\,\mathcal{D}(\bar{\mathbb{P}}_{\mathbf{Y}_w}\parallel \mathbb{P}_{\mathbf{Z}_w}),
\end{aligned}\end{equation}
where $\bar{\mathbb{P}}_{\mathbf{Y}_w}=\frac{1}{n}\sum_{i=1}^n\mathbb{P}_{\mathbf{Y}_{w,i}}$, (a) follows from the chain rule of entropy and removing conditioning, and (b) follows since the KL divergence, $\mathcal{D}(\mathbb{P}_{\mathbf{Y}_w}\parallel \mathbb{P}_{\mathbf{Z}_w})$, is convex in $\mathbb{P}_{\mathbf{Y}_w}$. Consequently, the constraint, $\mathcal{D}(\mathbb{P}_{\mathbf{Y}_w}\parallel \mathbb{P}_{\mathbf{Z}_w})\leq \frac{2\,\delta_{kl}^2}{n}$, is satisfied. Hence, the input distribution that maximizes the first-order approximation of the covert coding rate in \eqref{eq:converse1_prt1}, while minimizing $\mathcal{D}(\mathbb{P}_{\mathbf{Y}_w}\parallel \mathbb{P}_{\mathbf{Z}_w})$, in \eqref{eq:converse1_prt2}, is the zero-mean circularly symmetric complex Gaussian distribution with a covariance matrix $\mathbf{Q}_n$.
\end{IEEEproof}
 
\begin{IEEEproof} {\bf Achievability:} 

In this part, we show that if the scaling is less than the right side of \eqref{eq:thm1}, there exists a sequence of $(2^{{\lceil nR\rceil}},n,\epsilon,\delta_{kl})$-codes with a vanishing average probability of decoding error. 

\noindent{\bf 1- Choosing a sequence of distributions:} Alice chooses a sequence of input distributions $\left\{f_{\mathbf{X}}\right\}$, $\mathbf{X}\sim\mathcal{CN}({\bf 0},{\bf Q}_n)$, such that the sequence of output distributions $\left\{f_{\mathbf{Y}_{u}}\right\}$, $\mathbf{Y}_u\sim\mathcal{CN}({\bf 0},{\bf\Sigma}_{{\mathbf Y}_u}={\bf H}_u\,{\bf Q}_n\,{\bf H}_u ^\dag+{\bf\Sigma}_u)$, $u\in\left\{b,w\right\},$ satisfies an average probability of decoding error, $P^{(n)}_e$, at Bob and the constraint $\mathcal{D}(\mathbb{P}_{\mathbf{Y}_w}\parallel \mathbb{P}_{\mathbf{Z}_w})\leq \frac{2\,\delta_{kl}^2}{n}$ at Willie. 

\noindent{\bf 2- Codebook generation:} We generate a random codebook using i.i.d. $2^{\lceil nR\rceil}$ sequences, ${\bf x}^n(m)$, $m\in\mathcal{M}_n$, according to $f_{\mathbf{X}^n}(\mathbf{x}^n)=\prod_{i=1}^n f_{\mathbf{X}}(\mathbf{x}_i)$. The codebook is kept secret, according to Definition \ref{secret codebook}, between Alice and Bob.

\noindent{\bf 3- Average power constraint:} For an arbitrary small $\delta_{kl}\geq 0$, the average power constraint is not active. For reasoning about satisfying this constraint under the single-letter KL constraint, see the discussion in Subsection \ref{using the product distribution}. 

\noindent{\bf 4- Output distributions:} Since each codeword is generated randomly according to a product distribution $f_{\mathbf{X}^n}$ and the channel is memoryless according to \eqref{memoryless}, the output distribution is i.i.d. according to the product distribution $f_{\mathbf{Y}_u^n}(\mathbf{y}_u^n)=\prod_{i=1}^n f_{\mathbf{Y}_u}(\mathbf{y}_{u,i}),\ u\in\left\{b,w\right\}$. For reasoning about using the product distribution, see the discussion in Subsection \ref{using the product distribution}. 

\noindent{\bf 5- KL constraint at Willie:} The KL constraint is satisfied since the output distribution at Willie is i.i.d. Thus, $\mathcal{D}(\mathbb{P}_{\mathbf{Y}_w^n}\parallel \mathbb{P}_{\mathbf{Z}_w^n})=n\mathcal{D}(\mathbb{P}_{\mathbf{Y}_w}\parallel\mathbb{P}_{\mathbf{Z}_w})\leq 2\delta_{kl}^2$.

\noindent{\bf 6- Vanishing probability of decoding error at Bob:} We generalize the proof of the achievability of Theorem 1 in \cite{wang2016fundamental} for DMC channels to prove the remaining part of the achievability for MIMO AWGN channels. There exists a sequence of $(2^{{\lceil nR\rceil}},n,\epsilon,\delta_{kl})$-codes with a vanishing average probability of decoding error, if the scaling is less than the right side of \eqref{eq:thm1}. Hence, the sequence $\left\{\sqrt{n}R(n,\epsilon,\delta_{kl})\right\}$ is achievable if the random sequence $\left\{\frac{1}{\sqrt{n}}\mathbb{I}({\bf X}^n;{\bf Y}_b^n)\right\}$ converges to $\sqrt{n}\,C_{c}({\bf Q}_n)$ in probability, i.e., the following equation is satisfied:
\begin{equation}\label{eq:P_lim_inf}
\lim_{n \rightarrow \infty} Pr\left[\abs{\frac{1}{\sqrt{n}}\mathbb{I}({\bf X}^n;{\bf Y}^n_b)-\sqrt{n}\,C_{c}({\bf Q}_n)}\geq t\right]=0,\ t>0.
\end{equation}
 
To complete the achievability proof, we exploit Chebyshev’s inequality. We start by calculating $\frac{1}{\sqrt{n}}\mathbb{I}({\bf X}^n;{\bf Y}^n_b)$ and its moments as follows: Let $\mathbb{I}({\bf X}^n;{\bf Y}^n_b)=\mathbb{I}$
\begin{equation}\begin{aligned}\label{reviewer}
\mathbb{I}=&\log\frac{f_{\mathbf{Y}_{\it b}^n|\mathbf{X}^n}(\mathbf{y}_b^n|\mathbf{x}^n)}{f_{\mathbf{Y}_b^n}(\mathbf{y}_b^n)}= \log\prod_{i=1}^n\frac{ f_{\mathbf{Y}_{b}|\mathbf{X}}(\mathbf{y}_{b,i}|\mathbf{x}_i)}{f_{\mathbf{Y}_{b}}(\mathbf{y}_{b,i})}\\
=&\sum_{i=1}^n\,\log\frac{ f_{\mathbf{Y}_{b}|\mathbf{X}}(\mathbf{y}_{b,i}|\mathbf{x}_i)}{f_{\mathbf{Y}_{b}}(\mathbf{y}_{b,i})}\\
=&\sum_{i=1}^n\,\log\abs{ \mathbf{\Sigma}_{\mathbf{Y}_b}\,{\mathbf{\Sigma}_b}^{-1}}\\
+&\sum_{i=1}^n \mathbf{y}_{b,i}^{ \dag} \,{\bf\Sigma}_{\mathbf{Y}_b}^{-1}\,\mathbf{y}_{b,i} -(\mathbf{y}_{b,i} -\mathbf{H}_b\,\mathbf{x}_i)^{ \dag}\, {\bf\Sigma}_b^{-1} (\mathbf{y}_{b,i} -\mathbf{H}_b\,\mathbf{x}_i),
\end{aligned}\end{equation}
where $$f_{\mathbf{Y}_{b}|\mathbf{X}}(\mathbf{y}_{b,i}|\mathbf{x}_i)=\frac{\exp\left(-(\mathbf{y}_{b,i} -\mathbf{H}_b\mathbf{x}_i)^{ \dag} {\bf\Sigma}_b^{-1}(\mathbf{y}_{b,i} -\mathbf{H}_b\mathbf{x}_i)\right)}{\abs{\pi\mathbf{\Sigma}_b}^{-1} },$$ and $$f_{\mathbf{Y}_{b}}(\mathbf{y}_{b,i})= \abs{\pi\,\mathbf{\Sigma}_{\mathbf{Y}_b}}^{-1} \exp\left(-\mathbf{y}_{b,i}^{ \dag} \,{\bf\Sigma}_{\mathbf{Y}_b}^{-1}\,\mathbf{y}_{b,i} \right).$$
Then,
\begin{equation}\begin{aligned}
\mathbb{E}\left[\mathbb{I}\right]=&\ \sum_{i=1}^n\,\log\abs{ \mathbf{\Sigma}_{\mathbf{Y}_b}\,{\mathbf{\Sigma}_b}^{-1}}\\
+&\ \sum_{i=1}^n\mathbb{E}\left[\left(\mathbf{tr}\left(\mathbf{\Sigma}^{-1}_{\mathbf{Y}_b}\mathbf{Y}_{b,i}\,\mathbf{Y}_{b,i}^{ \dag}\right)-\mathbf{tr}\left(\mathbf{\Sigma}_b^{-1}\,\mathbf{Z}_{b,i}\,\mathbf{Z}_{b,i}^{\dag}\right)\right)\right]\\
=&\ n\,\log\abs{ \mathbf{\Sigma}_{\mathbf{Y}_b}\,{\mathbf{\Sigma}_b}^{-1}}\\
+&\ \sum_{i=1}^n\left(\mathbf{tr}\left(\mathbf{\Sigma}^{-1}_{\mathbf{Y}_b}\,\mathbf{\Sigma}_{\mathbf{Y}_b}\right)-\mathbf{tr}\left(\mathbf{\Sigma}_b^{-1}\,\mathbf{\Sigma}_b\right)\right)\\
=&\ n\,\log\abs{\frac{1}{\sigma_b^2}\,{\bf H}_b\,{\bf Q}_n\,{\bf H}_b ^\dag+{\bf I}_{N_b}}= n\,C_{c}({\bf Q}_n),
\end{aligned}
\end{equation}
\begin{equation}\begin{aligned}
\mathbf{var}\left[\frac{1}{\sqrt{n}}\mathbb{I}\right]=&\ \frac{1}{n}\,\mathbf{var}\left[\sum_{i=1}^n\,\log\abs{ \mathbf{\Sigma}_{\mathbf{Y}_b}\,{\mathbf{\Sigma}_b}^{-1}}+\sum_{i=1}^n\,V_i\right]\\
=&\ \frac{1}{n}\,\sum_{i=1}^n\mathbf{var}\left[\log\abs{ \mathbf{\Sigma}_{\mathbf{Y}_b}\,{\mathbf{\Sigma}_b}^{-1}}+V_i\right]=\mathbb{E}\left[\,V^2\right],
\end{aligned}
\end{equation}
where, $ \forall i\in\left\{1,\dots,n\right\}$,
\begin{equation}\begin{aligned}
v_i=&\ (\mathbf{H}_b\,\mathbf{x}_i+\mathbf{z}_{b,i})^{ \dag}\,{\bf\Sigma}_{\mathbf{Y}_b}^{-1}\,(\mathbf{H}_b\,\mathbf{x}_i+\mathbf{z}_{b,i}) -\mathbf{z}_{b,i}^{ \dag}\,{\bf\Sigma}_b^{-1}\,\mathbf{z}_{b,i}\\
=&\ \mathbf{x}_i^{ \dag}\,\mathbf{H}_b^{ \dag}\, {\bf\Sigma}_{\mathbf{Y}_b}^{-1}\,\mathbf{H}_b\,\mathbf{x}_i+2\,\mathbf{x}_i^{ \dag}\,\mathbf{H}_b^{ \dag}\, {\bf\Sigma}_{\mathbf{Y}_b}^{-1}\,\mathbf{z}_{b,i}\\
+&\ \mathbf{z}_{b,i}^{ \dag}\,\left({\bf\Sigma}_{\mathbf{Y}_b}^{-1}-{\bf\Sigma}_b^{-1}\right)\mathbf{z}_{b,i}\\
\overset{(a)}{=}&\ \mathbf{x}_i^{ \dag}\,\mathbf{H}_b^{ \dag}\, {\bf\Sigma}_{\mathbf{Y}_b}^{-1}\,\mathbf{H}_b\,\mathbf{x}_i+2\,\mathbf{x}_i^{ \dag}\,\mathbf{H}_b^{ \dag}\, {\bf\Sigma}_{\mathbf{Y}_b}^{-1}\,\mathbf{z}_{b,i}\\
-&\ \frac{1}{\sigma^2_b}\,\mathbf{z}_{b,i}^{ \dag}\,\mathbf{H}_b\left(\mathbf{Q}_n^{-1}\,\sigma^2_b+\mathbf{H}_b^{ \dag}\,\mathbf{H}_b\right)^{-1}\,\mathbf{H}_b^{ \dag}\,\mathbf{z}_{b,i},
\end{aligned}
\end{equation}
and (a) follows by matrix inversion lemma. Since $\mathbf{Q}_n$ tends to $\mathbf{0}$ as $n$ goes to infinity, ${\bf x}_i$, $ \forall i$, vanishes to a zero vector and ${\bf\Sigma}_{\mathbf{Y}_b}^{-1}$ goes to $\sigma^2_b\,{\bf I}_{N_b}$. Hence, ${v}_i$, $ \forall i$, goes to zero and therefore, $\lim_{n \rightarrow \infty} \mathbf{var}\left[\frac{1}{\sqrt{n}}\mathbb{I}({\bf X}^n;{\bf Y}^n_b)\right]=0$ by the bounded convergence theorem. Consequently, using Chebyshev’s inequality: $\lim_{n \rightarrow \infty}Pr\left[\abs{\frac{1}{\sqrt{n}}\mathbb{I}({\bf X}^n;{\bf Y}^n_b)-\sqrt{n}\,C_{c}({\bf Q}_n)}\geq u\right]\leq\frac{1}{u^2}\lim_{n \rightarrow \infty} \mathbf{var}\left[\frac{1}{\sqrt{n}}\mathbb{I}({\bf X}^n;{\bf Y}^n_b)\right]=0$.
\end{IEEEproof}


\section{Power Allocation}\label{App:power allocation}

The power allocation for MIMO AWGN channels is known to be water-filling across non-zero eigen-directions of Bob's channel \cite{tse2005fundamentals,elgamal_kim_2011}. With the KL constraint, Alice can transmit at full power in the zero eigen-directions of Willie's channel if Bob is not aligned to any of these directions. Thus, we have three different transmission regimes:
\begin{itemize}
 \item Regime 1) Zero eigen-directions of Bob's channel: No transmission.

 \item Regime 2) Non-zero eigen-directions of Bob's channel and zero eigen-directions of Willie's channel: Full power transmission, i.e., the covert capacity is positive.
 
 \item Regime 3) Non-zero eigen-directions of both channels: KL-constrained transmission, i.e., the covert capacity is zero. 
 \end{itemize}
In this appendix, we provide an illustration of these regimes under the KL constraint via solving the power allocation problem in \eqref{eq:thm1}, which is maximizing the first-order approximation of the maximum covert coding rate for a finite blocklength. For an arbitrary number of transmit antennas, Alice can achieve a positive covert capacity in Regime 2 via exploiting the null-space of Willie's channel. When there is not any zero eigen-direction of Willie's channel as in Regime 3, the covert capacity is zero. To give insight into the power allocation problem, we consider the optimization problem in \eqref{eq:thm1} as follows:
\begin{equation}
\begin{aligned}
&\max_{\substack{\mathbf{Q}_{\it n} \succeq \mathbf{0} \\ 
\mathbf{tr}(\mathbf{Q}_{\it n}) \leq P}}
& & \mathrm{\log\abs{\frac{1}{\sigma_{\it b}^2}\,{\bf H}_{\it b}\,{\bf Q}_{\it n}\,{\bf H}_{\it b} ^\dag+{\bf I}_{N_{\it b}}}} \\
& \text{subject to:}
& &\mathcal{D}(\mathbb{P}_{\mathbf{Y}_w}\parallel\mathbb{P}_{\mathbf{Z}_w})\leq\frac{2\,\delta_{kl}^2}{n}.
\end{aligned}
\end{equation}

This problem is a convex optimization problem\footnote{The KL divergence $\mathcal{D}(\mathbb{P}_{\mathbf{Y}_w}\parallel\mathbb{P}_{\mathbf{Z}_w})$ is convex in $\mathbb{P}_{\mathbf{Y}_w}$, and hence, in $\mathbf{Q}_n$.}. Moreover, the Slater's condition holds if $\mathbf{Q}_n$ is chosen to be zero except the first diagonal element, $q_{1}$, is chosen to satisfy $\frac{q_{1}\,\abs{h_{w,1}}^2}{\sigma_w^2}-\log\left(\frac{q_{1}\,\abs{h_{w,1}}^2}{\sigma_w^2}+1\right)$ $\leq\frac{2\,\delta_{kl}^2}{n}$ where $h_{w,1}$ is the first diagonal element of ${\bf H}_w$. Hence, $\mathbf{Q}_n$ satisfies all constraints and the Slater's condition holds. Therefore, there is no duality gap and KKT conditions are necessary and sufficient for optimality. To write the Lagrangian function, we represent the single-letter KL divergence as follows: Let $\mathcal{D}(\mathbb{P}_{\mathbf{Y}_w}\parallel\mathbb{P}_{\mathbf{Z}_w})=\mathcal{D}$.
\begin{equation}\begin{aligned}\label{KL_divergance}
\mathcal{D}&= \mathbb{E}_{\mathbb{P}_{\mathbf{Y}_w}} \left[\log f_{\mathbf{Y}_w}({\mathbf{Y}_w}) - \log f_{\mathbf{Z}_w}({\mathbf{Z}_w}) \right] \\
&= \log{\abs{{\bf\Sigma}_{\mathbf{Y}_w}^{-1}{\bf\Sigma}_w}} +\mathbb{E}_{\mathbb{P}_{\mathbf{Y}_w}} \left[\mathbf{Z}_{w}^{ \dag}\, {\bf\Sigma}_w^{-1}\,\mathbf{Z}_{w} -\mathbf{Y}_{w}^{ \dag}\, {\bf\Sigma}_{\mathbf{Y}_w}^{-1}\,\mathbf{Y}_{w} \right] \\
&=-\log\abs{\frac{1}{\sigma_w^2}\,{\bf H}_w\,{\bf Q}_n\,{\bf H}_w ^\dag+{\bf I}_{N_w}}+\mathbf{tr}\left(\frac{1}{\sigma_w^2}\,{\bf H}_w\,{\bf Q}_n\,{\bf H}_w ^\dag\right),
\end{aligned} \end{equation}
where
\begin{equation}\begin{aligned}
\mathbb{E}_{\mathbb{P}_{\mathbf{Y}_w}}& \left[\mathbf{Z}_{w}^{ \dag}\, {\bf\Sigma}_w^{-1}\,\mathbf{Z}_{w} -\mathbf{Y}_{w}^{ \dag}\, {\bf\Sigma}_{\mathbf{Y}_w}^{-1}\,\mathbf{Y}_{w} \right]\\
&=\mathbb{E}_{\mathbb{P}_{\mathbf{Y}_w}}\left[\mathbf{tr}\left(\mathbf{\Sigma}_w^{-1}\,\mathbf{Z}_{w}\,\mathbf{Z}_{w}^{\dag}\right)-\mathbf{tr}\left(\mathbf{\Sigma}^{-1}_{\mathbf{Y}_w}\,\mathbf{Y}_{w}\,\mathbf{Y}_{w}^{ \dag}\right)\right]\\
&=\mathbf{tr}\left(\mathbf{\Sigma}_w^{-1}\,\mathbf{\Sigma}_{\mathbf{Y}_w}\right)-\mathbf{tr}\left(\mathbf{\Sigma}^{-1}_{\mathbf{Y}_w}\,\mathbf{\Sigma}_{\mathbf{Y}_w}\right)\\
&=\mathbf{tr}\left(\frac{1}{\sigma_w^2}\,{\bf H}_w\,{\bf Q}_n\,{\bf H}_w ^\dag+{\bf I}_{N_w}\right)-N_w\\
&=\mathbf{tr}\left(\frac{1}{\sigma_w^2}\,{\bf H}_w\,{\bf Q}_n\,{\bf H}_w ^\dag\right),
\end{aligned} \end{equation}
\begin{equation}\begin{aligned}
f_{\mathbf{Z}_w}({\mathbf{z}_w}) &= {\abs{\pi\,\mathbf{\Sigma}_w}^{-1} \exp\left(-\mathbf{z}_{w}^{ \dag}\, {\bf\Sigma}_w^{-1}\,\mathbf{z}_{w}\right)},\\
f_{\mathbf{Y}_w}({\mathbf{y}_w}) &= {\abs{\pi\,\mathbf{\Sigma}_{\mathbf{Y}_w}}^{-1} \exp\left(-\mathbf{y}_{w}^{ \dag}\, {\bf\Sigma}_{\mathbf{Y}_w}^{-1}\,\mathbf{y}_{w} \right)}.
\end{aligned}
\end{equation}

Accordingly, the Lagrangian function is given by:
\begin{equation}\label{lagrange_function}
\begin{aligned}
\mathcal{L}(\mathbf{Q}_n,\mu,\eta)=&\ \log \abs{\frac{1}{\sigma_{\it b}^2}\,{\bf Q}_n\,{\bf H}_{\it b} ^\dag\,{\bf H}_{\it b}+{\bf I}_{N_{\it a}}} - \mu\, (\mathbf{tr}(\mathbf{Q}_n) -P)\\
-&\ \eta \,\log\abs{\frac{1}{\sigma_w^2}\,{\bf Q}_n\,{\bf H}_w ^\dag\,{\bf H}_w+{\bf I}_{N_a}}\\+&\ \eta\,\mathbf{tr}\left(\frac{1}{\sigma_w^2}\,{\bf Q}_n\,{\bf H}_w ^\dag\,{\bf H}_w\right)- \frac{2\,\eta\,\delta_{kl}^2}{n}.
\end{aligned}
\end{equation}

From the KKT conditions, the Lagrangian function is optimal in $\mathbf{Q}_n^*$, i.e., $$\bigtriangledown_{\mathbf{Q}_n}\mathcal{L}(\mathbf{Q}_n,\mu^*,\eta^*)|_{\mathbf{Q}_n=\mathbf{Q}_n^*}=\mathbf{0}.$$ Hence, we calculate the gradient of the Lagrangian function and let it equal to a zero matrix of the same dimension as follows:
\begin{equation}\begin{aligned}\label{eq:primal_optimality}
\left[{\bf Q}_n^*+\sigma_{\it b}^2\left({\bf H}_{\it b} ^\dag\,{\bf H}_{\it b}\right)^{-1} \right]^{-1}&+ \eta \left[{\bf Q}_n^*+\sigma_{\it w}^2\left({\bf H}_{\it w} ^\dag\,{\bf H}_{\it w}\right)^{-1} \right]^{-1}\\
&- \mu\,\mathbf{I}_{N_a}- \frac{\eta}{\sigma_{\it w}^2}{\bf H}_{\it w} ^\dag\,{\bf H}_{\it w}=\mathbf{0}.
\end{aligned}\end{equation}

Using the generalized singular value decomposition as in \cite{Paige_GSVD}, there exist unitary matrices ${\bf V}_{u}\in \mathbb{C}^{N_u\,\times N_u}$, $u\in\left\{b,w\right\}$, ${\bf C}\in \mathbb{C}^{N_a\,\times N_a}$ and a non-singular upper triangular matrix ${\bf R}\in \mathbb{C}^{S\,\times S}$, such that 
\begin{equation}\begin{aligned}\label{GSVD}
{\bf H}_{\it w}=&{\bf V}_{\it w}\,{\bf D}_{\it w}\,{\bf U},\\
{\bf H}_{\it b}=&{\bf V}_{\it b}\,{\bf D}_{\it b}\,{\bf U},
\end{aligned}
\end{equation}
where 
\begin{equation}
{\bf U}= \left[{\bf 0}\ {\bf R}\right]\,{\bf C}^\dag,
\end{equation}
and ${\bf D}_{u}=\mathbf{diag}(\lambda_{u,1}^{\frac{1}{2}},\dots,\lambda_{u,N}^{\frac{1}{2}})\in \mathbb{R}^{N_u\,\times N}$, $\lambda_{u,i}\geq 0,\ \forall i\in\left\{1,\dots,N\right\}$. Consequently, the gradient of the Lagrangian function can be rewritten as:
\begin{equation}
\begin{aligned}
\label{eq:projection_to_Hw}
\left[\tilde{\bf Q}_n^*+\sigma_{\it b}^2{\bf \Lambda}_{\it b}^{-1}\right]^{-1}+ \eta \left[\tilde{\bf Q}_n^*+\sigma_{\it w}^2{\bf \Lambda}_{\it w}^{-1}\right]^{-1}
= \mu\mathbf{I}_{N_a}+\frac{\eta}{\sigma_{\it w}^2}{\bf \Lambda}_{\it w},
\end{aligned}
\end{equation}
where $\tilde{\bf Q}_n^*={\bf U}\,{\bf Q}_n^*\,{\bf U}^\dag$ and ${\bf \Lambda}_{u}={\bf D}_{u}^T\,{\bf D}_{u}=\mathbf{diag}(\lambda_{u,1},\dots,\lambda_{u,N})\in \mathbb{R}^{N\,\times N}$, $\lambda_{u,i}\geq 0,\ \forall i\in\left\{1,\dots,N\right\}$. Since $\eta$ is non-negative and all matrices in \eqref{eq:projection_to_Hw} are positive semi-definite, $\tilde{\bf Q}_n^*$ is a diagonal matrix. Also, by Hadamard's inequality, ${\tilde{\mathbf{Q}}_n}^*$ is a diagonal matrix since the KL constraint is an average power constraint \cite{wang2016fundamental}. Hence, the gradient in \eqref{eq:projection_to_Hw} can be decomposed as follows: 
\begin{equation}
\begin{aligned}
\label{eq:projection_to_Hw_scalar}
&\left[q^*_i+\frac{\sigma_{\it b}^2}{{\lambda}_{b,i}}\right]^{-1}+ \eta \left[q^*_i+\frac{\sigma_{\it w}^2}{{\lambda}_{w,i}}\right]^{-1}= \mu +\frac{\eta\,{\lambda}_{\it w,i}}{\sigma_{\it w}^2},
\end{aligned}
\end{equation}
$\forall i\in\left\{1,\dots,N\right\}$. Therefore, the optimal power allocation in each eigen-direction, $q^*_i,\ \forall i\in\left\{1,\dots,N\right\}$, is given by solving \eqref{eq:projection_to_Hw_scalar} as follows:
\begin{subnumcases}{q^*_i=}
 \frac{1}{2}\left[g_i+h_i\right]^+, & ${\lambda}_{w,i}\neq 0\ \mbox{and}\ {\lambda}_{b,i}\neq 0,$ \label{powerAllocation1}
 \\
 \left[\frac{1}{\mu}-\frac{\sigma_{\it b}^2}{{\lambda}_{b,i}}\right]^+, & ${\lambda}_{w,i}= 0\ \mbox{and}\ {\lambda}_{b,i}\neq 0,$ \label{powerAllocationNull}\\
 0, & $\mbox{otherwise},$\label{NoTransmission} 
\end{subnumcases}
$\forall i\in\left\{1,\dots,N\right\}$, where
\begin{equation}
 g_i=\sqrt{\left(\frac{\sigma_{\it w}^2\left(\eta+1\right)}{\sigma_{\it w}^2\,\mu +\eta\,{\lambda}_{\it w,i}}-a_i\right)^2+\frac{4a_i\sigma_{\it w}^2}{\sigma_{\it w}^2\,\mu +\eta\,{\lambda}_{\it w,i}}},
\end{equation}
\begin{equation}
    h_i=\frac{\sigma_{\it w}^2\left(\eta+1\right)}{\sigma_{\it w}^2\,\mu +\eta\,{\lambda}_{\it w,i}}-\left(\frac{\sigma_{\it b}^2}{{\lambda}_{b,i}}+\frac{\sigma_{\it w}^2}{{\lambda}_{w,i}}\right),
\end{equation}
and
\begin{equation}\label{a_i}
 a_i=\frac{\sigma_{\it w}^2}{{\lambda}_{w,i}}-\frac{\sigma_{\it b}^2}{{\lambda}_{b,i}},
\end{equation}

In addition, $\mu$ and $\eta$ are chosen such that they are feasible, i.e., $\mu\geq 0$ and $\eta\geq 0$, and satisfy the complementary slackness conditions: 1) $\mu^* \left(\mathbf{tr}(\mathbf{Q}_n^*) -P\right) = 0$ and 2) $\eta^*\left({\mathcal D}(\mathbf{Q}_n^*)-\frac{2\,\delta_{kl}^2}{n}\right) = 0$. Thus,
\begin{equation} \label{mu_power_allocation}
\sum_{i=1}^{N}q_i^*
\begin{cases}
=P,&\mu>0,\\
<P,&\mu=0,
\end{cases}
\end{equation}
and
\begin{equation}\label{eta_power_allocation}
\sum_{i=1}^{N}\left[{\frac{q^*_i\,{\lambda}_{w,i}}{\sigma_w^2}}-\log{\left(\frac{q^*_i\,{\lambda}_{w,i}}{\sigma_w^2}+1\right)}\right]
\begin{cases}
=\frac{2\,\delta_{kl}^2}{n},&\eta>0,\\
<\frac{2\,\delta_{kl}^2}{n},& \eta=0,
\end{cases}
\end{equation}
where the KL divergence in the right-hand side is calculated in Appendix \ref{App:KL_divergence}.

In \eqref{NoTransmission}, there is no transmission since $\lambda_{b,i}=0$, which represents Regime 1. In the meanwhile, the full power transmission is given by water-filling across the non-zero eigen-directions of Bob's channel as in \eqref{powerAllocationNull} when the corresponding eigen-directions of Willie's channel are zero, i.e., a null-space exists. This case represents Regime 2. Besides, when the eigen-directions of both channels are non-zero, the covert capacity is zero and the power allocation, in this case, is given in \eqref{powerAllocation1}, which represents Regime 3.

\section{Calculating the Single-Letter KL Divergence}\label{App:KL_divergence}
Using \eqref{KL_divergance}, the single-letter KL divergence is calculated as follows: Let $\mathcal{D}(\mathbb{P}_{\mathbf{Y}_w}\parallel\mathbb{P}_{\mathbf{Z}_w})=\mathcal{D}$.
\begin{equation}\begin{aligned}
\mathcal{D}&=\mathbf{tr}\left(\frac{1}{\sigma_w^2}\,{\bf H}_w\,{\bf Q}_n\,{\bf H}_w ^\dag\right)-\log\abs{\frac{1}{\sigma_w^2}\,{\bf H}_w\,{\bf Q}_n\,{\bf H}_w ^\dag+{\bf I}_{N_w}}\\
&=\sum_{i=1}^{N}\left[{\frac{q_i\,{\lambda}_{w,i}}{\sigma_{\it w}^2}}-\log\left(\frac{q_i\,{\lambda}_{w,i}}{\sigma_{\it w}^2}+1\right)\right],
\end{aligned} \end{equation}
where 
\begin{equation}\begin{aligned}
\log&\abs{\frac{1}{\sigma_w^2}\,{\bf H}_w\,{\bf Q}_n\,{\bf H}_w ^\dag+{\bf I}_{N_w}}\\
&{\overset{(a)}{=}\log\abs{\frac{1}{\sigma_w^2}\,{\bf Q}_n\,{\bf H}_w ^\dag\,{\bf H}_w+{\bf I}_{N_a}}}\\
&{{=}\log\abs{\frac{1}{\sigma_w^2}\,{\bf Q}_n\,{\bf U}^\dag\,{\bf D}_{\it w}^T\,{\bf V}_{\it w}^\dag\,{\bf V}_{\it w}\,{\bf D}_{\it w}\,{\bf U}+{\bf I}_{N_a}}}\\
&{\overset{(b)}{=}\log\abs{\frac{1}{\sigma_w^2}\,{\bf Q}_n\,{\bf U}^\dag\,{\bf \Lambda}_{\it w}\,{\bf U}+{\bf I}_{N_a}}}\\
&{\overset{(c)}{=}\log\abs{\frac{1}{\sigma_w^2}\,{\bf U}\,{\bf Q}_n\,{\bf U}^\dag\,{\bf \Lambda}_{\it w}+{\bf I}_{N_a}}}\\
&{{=}\log\abs{\frac{1}{\sigma_w^2}\,\tilde{\bf Q}_n\,{\bf \Lambda}_{\it w}+{\bf I}_{N_a}}}\\
&{\overset{(d)}{=}\sum_{i=1}^{N}\log\left(\frac{q_i\,{\lambda}_{w,i}}{\sigma_{\it w}^2}+1\right)},
\end{aligned} \end{equation}
where (a) follows since $\abs{{\bf A}{\bf B}+{\bf I}}=\abs{{\bf B}{\bf A}+{\bf I}}$, (b) since ${\bf V}_{\it w}$ is a unitary matrix, and ${\bf D}_{\it w}^T\,{\bf D}_{\it w}={\bf \Lambda}_{\it w}=\mathbf{diag}(\lambda_{w,1},\dots,\lambda_{w,N_a})\in \mathbb{R}^{N_a\,\times N_a}$, (c) since $\abs{{\bf A}{\bf B}+{\bf I}}=\abs{{\bf B}{\bf A}+{\bf I}}$, and (d) since $\tilde{\bf Q}_n$ is a diagonal matrix derived in Appendix \ref{App:power allocation}. Similarly,
\begin{equation}\begin{aligned}
\mathbf{tr}\left(\frac{1}{\sigma_w^2}\,{\bf H}_w\,{\bf Q}_n\,{\bf H}_w ^\dag\right)&=\mathbf{tr}\left(\frac{1}{\sigma_w^2}\,\tilde{\bf Q}_n\,{\bf \Lambda}_{\it w}\right)=\sum_{i=1}^{N}{\frac{q_i\,{\lambda}_{w,i}}{\sigma_{\it w}^2}}.
\end{aligned} \end{equation}

\section{Proof of Theorem \ref{thm:Scaling}}\label{App:proof_thm2}

\begin{IEEEproof}{\bf Converse:} Consider a given $(2^{\lceil nR\rceil},n,\epsilon,\delta_{kl})$-code that satisfies the KL constraint, $\mathcal{D}(\mathbb{P}_{\mathbf{Y}_w^n}\parallel \mathbb{P}_{\mathbf{Z}_w^n})\leq 2\,\delta_{kl}^2$, with an average probability of decoding error  $P^{(n)}_e$ such that $\lim_{n \rightarrow \infty} P^{(n)}_e=0$. From the converse of Theorem 1, $nR(n,\epsilon,\delta_{kl})\leq\,n\,C_{c}({\bf Q}_n)+n\,\epsilon_n$, and $\mathcal{D}(\mathbb{P}_{\mathbf{Y}_w}\parallel\mathbb{P}_{\mathbf{Z}_w})\leq\frac{2\,\delta_{kl}^2}{n}$. Hence, the KL divergence can be lower bounded as follows:
\begin{equation}\begin{aligned}\label{bound_on_KL}
\frac{2\,\delta_{kl}^2}{n}&\geq\mathcal{D}(\mathbb{P}_{\mathbf{Y}_w}\parallel\mathbb{P}_{\mathbf{Z}_w})\\
&=\sum_{i=1}^{N}\left[{\frac{q_i\,{\lambda}_{w,i}}{\sigma_{\it w}^2}}-\log\left(\frac{q_i\,{\lambda}_{w,i}}{\sigma_{\it w}^2}+1\right)\right]\\
&{\geq}\sum_{i=1}^{N}{\frac{q_i^2\,{\lambda}_{w,i}^2}{2\,\sigma_{\it w}^4}},
\end{aligned} \end{equation}
where the second inequality follows from the logarithm fact that $\log(x+1)\geq x-\frac{x^2}{2},\ x>0$. Let $\frac{q_i^2\,\lambda_{w,i}^2}{2\,\sigma_{\it w}^4}\leq\frac{2\,c_i\,\delta_{kl}^2}{n}, \ \forall i$ such that the bound in \eqref{bound_on_KL} is satisfied when summing over the available degrees of freedom. Consequently, 
\begin{equation}
q_i\leq\frac{2\,\sigma_{\it w}^2}{{\tilde{\lambda}_{w,i}}}\sqrt{\frac{c_i\,\delta_{kl}^2}{n}},\ \forall i.
\end{equation}

Therefore, an upper bound on the scaling, $L$, can be obtained as follows:
\begin{equation}
\begin{aligned}
L\triangleq& \lim_{\epsilon \downarrow 0} \varliminf_{n \rightarrow \infty} \sqrt{\frac{n}{2\,\delta_{kl}^2}}{R}(n,\epsilon,\delta_{kl})\\
\leq&\varliminf_{n \rightarrow \infty} \sqrt{\frac{n}{2\,\delta_{kl}^2}}\, C_{c}({\bf Q}_n)\\
\leq&\sum_{i=1}^{N}\frac{\sqrt{2\,c_i}\,\sigma_w^2\,{\lambda}_{\it b,i}}{\sigma_{\it b}^2\,\lambda_{w,i}},
\end{aligned}\end{equation}
where
\begin{equation}
\begin{aligned}
C_{c}({\bf Q}_n)=& \log\abs{\frac{1}{\sigma_{\it b}^2}\,{\bf H}_{\it b}\,{\bf Q}_n\,{\bf H}_{\it b} ^\dag+{\bf I}_{N_{\it b}}}\\
=& \log\abs{\frac{1}{\sigma_{\it b}^2}\,\tilde{\bf Q}_n\,{\bf \Lambda}_{\it b}+{\bf I}_{N_{\it a}}}\\
=& \sum_{i=1}^{N}\log\left(\frac{q_i\,{\lambda}_{\it b,i}}{\sigma_{\it b}^2}+1\right)\\
\stackrel{(a)}{\leq}& \sqrt{\frac{n}{2\,\delta_{kl}^2}} \sum_{i=1}^{N}\frac{q_i\,{\lambda}_{\it b,i}}{\sigma_{\it b}^2}\\
\stackrel{(b)}{\leq}& \sum_{i=1}^{N}\frac{2\,\sigma_w^2\,{\lambda}_{\it b,i} }{\sigma_{\it b}^2\,{\lambda}_{w,i}}\sqrt{\frac{c_i\,\delta_{kl}^2}{n}},
\end{aligned}\end{equation}(a) follows from the logarithm inequality and (b) by substituting an upper bound on $q_i,\ \forall i$.

{\bf Achievability:} The KL divergence can be upper bounded as follows:
\begin{equation}\begin{aligned}
\mathcal{D}(\mathbb{P}_{\mathbf{Y}_w}\parallel\mathbb{P}_{\mathbf{Z}_w})&=\sum_{i=1}^{N}\left[{\frac{q_i\,{\lambda}_{w,i}}{\sigma_{\it w}^2}}-\log\left(\frac{q_i\,{\lambda}_{w,i}}{\sigma_{\it w}^2}+1\right)\right]\\
&\stackrel{(a)}{\leq} \sum_{i=1}^{N} \left[{\frac{q_i\,{\lambda}_{w,i}}{\sigma_{\it w}^2}}-\left(\frac{\frac{2\,q_i\,{\lambda}_{w,i}}{\sigma_{\it w}^2}}{2+\frac{q_i\,{\lambda}_{w,i}}{\sigma_{\it w}^2}}\right)\right]\\
&= \sum_{i=1}^{N}{\frac{q_i^2\,{\lambda}^2_{w,i}}{2\,\sigma_{\it w}^4+q_i\,{\lambda}_{w,i}\sigma_{\it w}^2}}\leq\frac{2\,\delta_{kl}^2}{n},
\end{aligned} \end{equation}
where (a) follows from the logarithm inequality, $\frac{2x}{2+x}\leq\log(1+x),\ \forall x\geq 0$. Let ${\frac{q_i^2\,{\lambda}^2_{w,i}}{2\,\sigma_{\it w}^4+q_i\,{\lambda}_{w,i}\sigma_{\it w}^2}}\leq\frac{2\,c_i\,\delta_{kl}^2}{n}, \ \forall i,$ then,
\begin{equation}
q_i\leq\frac{c_i\sigma_w^2\,\delta_{kl}^2 }{n\,{\lambda}_{w,i}}\left[\sqrt{1+\frac{4\,n}{c_i\,\delta_{kl}^2}}+1\right], \ \forall i. 
\end{equation}

Hence, choosing $q_i=\frac{2\,\sigma_w^2}{{\lambda}_{w,i}}\sqrt{\frac{c_i\,\delta_{kl}^2}{n}},\ \forall i$, satisfies the KL constraint. Therefore, from the achievability of Theorem 1 and using \eqref{eq:P_lim_inf}, the following bound on the scaling, $L$, is achievable.
\begin{equation}
\begin{aligned}
L\geq& \varliminf_{n \rightarrow \infty}\frac{1}{\sqrt{2\,n\,\delta_{kl}^2}}\mathbb{E}\left[\mathbb{I}({\bf X}^n;{\bf Y}^n_b)\right]\\
=& \varliminf_{n \rightarrow \infty}\sqrt{\frac{n}{2\,\delta_{kl}^2}}\,C_{c}({\bf Q}_n)
\geq\sum_{i=1}^{N}\frac{\sqrt{2\,c_i}\,\sigma_w^2\,{\lambda}_{\it b,i}}{\sigma_{\it b}^2\,{\lambda}_{w,i}},
\end{aligned} 
\end{equation}
where
\begin{equation}
\begin{aligned}
C_{c}({\bf Q}_n)=& \log\abs{\frac{1}{\sigma_{\it b}^2}\,{\bf H}_{\it b}\,{\bf Q}_n\,{\bf H}_{\it b} ^\dag+{\bf I}_{N_{\it b}}}\\
=& \log\abs{\frac{1}{\sigma_{\it b}^2}\,\tilde{\bf Q}_n\,{\bf \Lambda}_{\it b}+{\bf I}_{N_{\it a}}}\
= \sum_{i=1}^{N}\log\left(\frac{q_i\,{\lambda}_{\it b,i}}{\sigma_{\it b}^2}+1\right)\\
\stackrel{(a)}{\geq}&\sum_{i=1}^{N}\left(\frac{\frac{2\,q_i\,{\lambda}_{\it b,i}}{\sigma_{\it b}^2}}{\frac{q_i\,{\lambda}_{\it b,i}}{\sigma_{\it b}^2}+2}\right)\\
\stackrel{(b)}{\geq}&\sum_{i=1}^{N}\left(\frac{{4\,\sigma_w^2\,{\lambda}_{\it b,i}\,\sqrt{{c_i\,\delta_{kl}^2}}}}{{2\,{\sigma_w^2}\,{\lambda}_{\it b,i}\,\sqrt{{c_i\,\delta_{kl}^2}}}+2\,\sqrt{n}\,{\lambda}_{w,i}\,\sigma_{\it b}^2}\right),
\end{aligned} 
\end{equation}
(a) follows from the logarithm inequality and (b) holds by substituting the chosen value of $q_i,\ \forall i$.
\end{IEEEproof}


\section{Proof Sketch for Corollary \ref{thm:sec_covert}}\label{App:proof_corollary5}
\begin{IEEEproof} The key point is that the KL constraint forces the average power to decay with the blocklength. Note that, this is identical to the situation in \cite{wang2016fundamental}. Thus, in the achievability proof, the sequence of input distributions can be interpreted as an input distribution with a sequence of covariance matrices. In the following, we sketch the achievability proof. Alice chooses a sequence of input distributions such that the sequence of output distributions satisfies an average probability of decoding error at Bob and the single-letter KL constraint at Willie. For every blocklength, we generate a random codebook according to the chosen distribution. The codebook is not kept secret. For an arbitrary small $\delta_{kl}\geq 0$, the average power constraint and the KL constraint are satisfied. 

From the achievability proof of Theorem \ref{thm:Kn_covert}, the random sequence $\left\{\frac{1}{\sqrt{n}}\mathbb{I}({\bf X}^n;{\bf Y}_b^n)\right\}$ converges to $\sqrt{n}\,R_{b}({\bf Q}_n^*)$ in probability. Thus, the sequence $\left\{\sqrt{n}R(n,\epsilon,\delta_{kl})\right\}$ is achievable, i.e., there exists a sequence of $(2^{{\lceil nR\rceil}},n,\epsilon,\delta_{kl})$-codes with a vanishing average probability of decoding error. Since the codebook is public, similarly, the random sequence $\left\{\frac{1}{\sqrt{n}}\mathbb{I}({\bf X}^n;{\bf Y}_w^n)\right\}$ converges to $\sqrt{n}\,R_{w}({\bf Q}_n^*)$ in probability. Thus, it is possible for Willie to decode and get information $\mathcal{O}(\sqrt{n})$ bits with a vanishing average probability of decoding error. That is, the information leakage to Willie is $\mathrm{I}(M;{\bf Y}^n_w)$, which scales as $\mathcal{O}(\sqrt{n})$ bits, even when the information rate to both Bob and Willie is zero due to the KL constraint. This means that Willie is still able to detect the presence of the communication session by directly decoding the message. This implies that the KL constraint is not sufficient without IS. For DMC, the amount of information leaked to Willie scales as $\mathcal{O}(\sqrt{n})$, and the decay of the information leakage with the blocklength is characterized in \cite[Theorem 5]{Bloch_resolvability}.

To prevent Willie from decoding the message, Alice exploits a stochastic encoder with a sequence of randomization rates greater than $\left\{\frac{1}{{n}}\mathbb{I}({\bf X}^n;{\bf Y}_w^n)\right\}$. This overwhelms Willie with a sequence $\left\{2^{\mathbb{I}({\bf X}^n;{\bf Y}_w^n)}\right\}$ of dummy messages that are $\mathcal{O}(\sqrt{n})$ bits due to the KL constraint. On the other hand, Bob is still able to decode a sequence $\left\{2^{\mathbb{I}({\bf X}^n;{\bf Y}_b^n)-\mathbb{I}({\bf X}^n;{\bf Y}_w^n)}\right\}$ of confidential messages that are $\mathcal{O}(\sqrt{n})$ bits as well.
\end{IEEEproof}

\section{Proof of Theorem \ref{thm:secrecy_Scaling}}\label{App:proof_thm3}

\begin{IEEEproof}
{\bf Converse:} Similar to the converse of Theorem 2, an upper bound on the scaling is obtained as follows:
\begin{equation}
\begin{aligned}\label{L_Ssecrecy_converse}
L_{S}&\triangleq \lim_{\epsilon \downarrow 0} \varliminf_{n \rightarrow \infty} \sqrt{\frac{n}{2\,\delta_{kl}^2}}{R}(n,\epsilon,\delta_{kl},\delta_{s})\\
&\leq\varliminf_{n \rightarrow \infty}\sqrt{\frac{n}{2\,\delta_{kl}^2}} \,C_s({\bf Q}_n)\leq\sum_{i=1}^{N}\sqrt{2\,c_i}\left[\frac{\,{\sigma_w^2}\,{\lambda}_{\it b,i}}{\sigma_{\it b}^2\,{\lambda}_{w,i}}-1\right]^+,
\end{aligned} 
\end{equation}
where
\begin{equation}
\begin{aligned}
C_s&= \left[\log\abs{\frac{1}{\sigma_{\it b}^2}\,\tilde{\bf Q}_n\,{\bf \Lambda}_{\it b}+{\bf I}_{N_{\it a}}}-\log\abs{\frac{1}{\sigma_w^2}\,\tilde{\bf Q}_n\,{\bf \Lambda}_{\it w}+{\bf I}_{N_{\it a}}}\right]^+\\
&= \sum_{i=1}^{N}\left[\log\left(\frac{q_i\,{\lambda}_{\it b,i}}{\sigma_{\it b}^2}+1\right)-\log\left(\frac{q_i\,{\lambda}_{\it w,i}}{\sigma_{\it w}^2}+1\right)\right]^+\\
&\stackrel{(a)}{\leq} \sum_{i=1}^{N}\left[\frac{q_i\,{\lambda}_{\it b,i}}{\sigma_{\it b}^2}-\left(\frac{\frac{q_i\,{\lambda}_{\it w,i}}{\sigma_{\it w}^2}}{\frac{q_i\,{\lambda}_{\it w,i}}{\sigma_{\it w}^2}+1}\right)\right]^+\\
&\stackrel{(b)}{\leq} \sum_{i=1}^{N}\left[\frac{2\,{\sigma_w^2}\,{\lambda}_{\it b,i}}{\sigma_{\it b}^2\,{\lambda}_{w,i}}\sqrt{\frac{c_i\,\delta_{kl}^2}{n}}-\left(\frac{2\,\sqrt{\frac{c_i\,\delta_{kl}^2}{n}}}{2\,\sqrt{\frac{c_i\,\delta_{kl}^2}{n}}+1}\right)\right]^+,
\end{aligned} 
\end{equation}
(a) follows from the logarithm inequality and (b) holds by substituting the chosen value of $q_i,\ \forall i$.

{\bf Achievability:} Using the achievability of Theorem 2, the following bound on the scaling law is achievable:
\begin{equation}
\begin{aligned}
\label{L_Ssecrecy_ach}
L_{S}&\geq \varliminf_{n \rightarrow \infty}\frac{1}{\sqrt{2\,n\,\delta_{kl}^2}}\mathbb{E}\left[\mathbb{I}({\bf X}^n;{\bf Y}^n_b)-\mathbb{I}({\bf X}^n;{\bf Y}^n_w)\right]^+\\
&=\varliminf_{n \rightarrow \infty}\,C_s({\bf Q}_n)\geq\sum_{i=1}^{N}\sqrt{2\,c_i}\left[\frac{\,{\sigma_w^2}\,{\lambda}_{\it b,i}}{\sigma_{\it b}^2\,{\lambda}_{w,i}}-1\right]^+,
\end{aligned} 
\end{equation}
where
\begin{equation}
\begin{aligned}
C_s&=\sum_{i=1}^{N}\left[\log\left(\frac{q_i\,{\lambda}_{\it b,i}}{\sigma_{\it b}^2}+1\right)-\log\left(\frac{q_i\,{\lambda}_{\it w,i}}{\sigma_{\it w}^2}+1\right)\right]^+\\
& \stackrel{(a)}{\geq}\sum_{i=1}^{N}\left[\frac{\frac{q_i\,{\lambda}_{\it b,i}}{\sigma_{\it b}^2}}{\frac{q_i\,{\lambda}_{\it b,i}}{\sigma_{\it b}^2}+1}-\frac{q_i\,{\lambda}_{\it w,i}}{\sigma_{\it w}^2}\right]^+\\
&\stackrel{(b)}{\geq}\sum_{i=1}^{N}\left[\frac{{2{\sigma_w^2}{\lambda}_{\it b,i}}\sqrt{{c_i\delta_{kl}^2}}}{{2{\sigma_w^2}{\lambda}_{\it b,i}}\sqrt{{c_i\delta_{kl}^2}}+\sigma_{\it b}^2{\lambda}_{w,i}\sqrt{n}}-2\sqrt{\frac{c_i\delta_{kl}^2}{n}}\right]^+,
\end{aligned} 
\end{equation}
(a) follows from the logarithm inequality and (b) holds by substituting the chosen value of $q_i,\ \forall i$.
\end{IEEEproof}


\section{The Scaling in the Covert Compound Channel Setting} \label{App:compound}

In this appendix, we show, 1) under the KL constraint, we can still choose appropriate quantization levels to guarantee that the maximum probability of decoding error (the maximum overall number of states) of the compound covert DMC with an arbitrary uncertainty set vanishes when using a code designed for an approximated compound covert DMC with a finite uncertainty set; 2) the maximum covert rate of the compound MIMO AWGN channel scales as the covert rate of the worst channel within the underlying class of channels. We follow closely the same arguments that are introduced in \cite{schaefer2015secrecy} for the compound DMC setting given an arbitrary channel uncertainty set and under the strong IS constraint. Although it is natural to start with covert communication with a secret codebook, we directly and briefly investigate a specific case when the codebook is not kept secret and a strong IS constraint is used. In the following, we give related results under both the KL constraint and the strong IS constraint.

{\bf 1- A compound covert DMC with IS and a finite uncertainty set:} Let $\mathcal{X}$ be a finite input set, $\mathcal{Y}_b,\ \mathcal{Y}_w$ be finite output sets, and $\mathcal{S}$ be the channel uncertainty set. For every state $s\in \mathcal{S}$, Bob's and Willie's channels transition are given by $\mathbb{W}^n_s(y^n_b|x^n)=\Pi^n_{i=1}\mathbb{W}_s(y_{b,i}|x_i)$ and $\mathbb{V}^n_s(y^n_w|x^n)=\Pi^n_{i=1}\mathbb{V}_s(y_{w,i}|x_i)$, respectively, while $x^n\in \mathcal{X}^n$ and $y_u^n\in \mathcal{Y}_u^n,\ u\in\left\{b,w\right\},$ are input and output sequences. The input distribution $\mathbb{P}_n$ induces output distributions $\mathbb{Q}_{Y_b}$ and $\mathbb{Q}_{Y_w}$ at Bob and Willie respectively. Then, the compound covert DMC with IS is given by $\mathfrak{M}=\left\{\left(\mathbb{W}_s,\mathbb{V}_s\right):\ s\in\mathcal{S} \right\}$. We consider a $(2^{{\lceil nR\rceil}},n,\epsilon,\delta_{kl},\delta_{s})$-code in Definition \ref{code_IS} but defined for a covert DMC with IS under 1) a maximum probability of decoding error, 2) a maximum KL divergence, and 3) a maximum information leakage, where the maximum is over all states. For a fixed state $s\in\mathcal{S}$, the average probability of decoding error decays as follows \cite[Theorem 5]{Bloch_resolvability}:
\begin{equation}
 P_e^{(n)}\leq 2^{-\xi_1\,\omega_n\sqrt{n}},
\end{equation}
where $\omega_n=o(1)\cap \omega(1/\sqrt{n})$ such that $\lim_{n\rightarrow\infty}\omega_n=0$ and $\lim_{n\rightarrow\infty}\omega_n\sqrt{n}=\infty$. In the meanwhile, the KL divergence and information leakage decay as follows:
\begin{equation}\begin{aligned}\label{information_leakage_decay_DMC}
 \mathcal{D}(\mathbb{Q}_{Y_w^n}\parallel \mathbb{Q}_{Z_w^n})&\leq \mathcal{D}_n+ 2^{-\xi_2\,\omega_n\sqrt{n}},\\
 I(\mathbb{P}_M,\mathbb{V}^n_s)&\leq 2^{-\xi_3\,\omega_n\sqrt{n}},
\end{aligned}
\end{equation}
where $\mathcal{D}_n\rightarrow\infty$ as $n\rightarrow\infty$, and $\xi_i>0,\ \forall i,$ that depend on $\mathbb{W}_s$ and $\mathbb{V}_s$. Using \cite[Theorem 2]{compound_results} for compound wiretap DMC and \cite{wang2016fundamental} for covert communication over DMC, one can derive the following. For a finite uncertainty set, Alice can choose a sequence of input distributions such that there exists a sequence of $(2^{{\lceil nR\rceil}},n,\epsilon,\delta_{kl},\delta_{s})$-codes that achieves the following scaling of the compound covert DMC with IS:
\begin{equation}
\begin{aligned}
L_c\geq\varliminf_{n \rightarrow \infty} \sqrt{\frac{n}{2\,\delta_{kl}^2}}&\max_{\mathbb{P}_n}
& & {\left[\min_{s\in\mathcal{S}}I(\mathbb{P}_n;\mathbb{W}_s)-\max_{s\in\mathcal{S}}I(\mathbb{P}_n;\mathbb{V}_s)\right]} \\
& \text{s. t.}
& &\max_{s\in\mathcal{S}}\mathcal{D}(\mathbb{Q}_{Y_w}\parallel \mathbb{Q}_{Z_w})\leq\frac{2\,\delta_{kl}^2}{n},
\end{aligned}
\end{equation}
with a maximum probability of decoding error, a maximum KL divergence and a maximum information leakage that decay as follows:
\begin{equation}
 \begin{aligned}
\max_{s\in\mathcal{S}}P_e^{(n)}&\leq |\mathcal{S}|^{\frac{1}{4}}2^{-\xi_1\,\omega_n\sqrt{n}},\\
 \max_{s\in\mathcal{S}}\mathcal{D}(\mathbb{Q}_{Y_w^n}\parallel \mathbb{Q}_{Z_w^n})&\leq \mathcal{D}_n+ 2^{-\xi_2\,\omega_n\sqrt{n}},\\
 \max_{s\in\mathcal{S}} I(\mathbb{P}_M,\mathbb{V}^n_s)&\leq 2^{-\xi_3\,\omega_n\sqrt{n}},
\end{aligned}
\end{equation}
where $\xi_i>0,\ \forall i,$ that does not depend on a specific state. 

{\bf 2- Extension to the compound covert DMC with IS and an arbitrary uncertainty set:} Using Lemma 1 and Lemma 2 in \cite{schaefer2015secrecy}, for every integer $\ell_q\geq 2|\mathcal{Y}_b|^2|\mathcal{Y}_w|^2$, there exists a compound covert DMC with IS that is given by $\mathfrak{M}_{\ell_q}=\left\{\left(\bar{\mathbb{W}}_s,\bar{\mathbb{V}}_s\right):\ s\in\mathcal{S}_{\ell_q} \right\}$ with a finite uncertainty set, and $|\mathcal{S}_{\ell_q}|\leq ({\ell_q}+1)^{|\mathcal{X}||\mathcal{Y}_b||\mathcal{Y}_w|}$, such that any pair $\left({\mathbb{W}}_s,{\mathbb{V}}_s\right)\in\mathfrak{M}$ can be closely approximated by a pair $\left(\bar{\mathbb{W}}_s,\bar{\mathbb{V}}_s\right)\in\mathfrak{M}_{\ell_q}$. Moreover, any sequence of $(2^{{\lceil nR\rceil}},n,\epsilon,\delta_{kl},\delta_{s})$-codes for $\left(\bar{\mathbb{W}}_s,\bar{\mathbb{V}}_s\right)$ can be used for $\left({\mathbb{W}}_s,{\mathbb{V}}_s\right)$ with the following decay rates:
\begin{equation}
\begin{aligned}
\max_{s\in\mathcal{S}}P_e^{(n)}\leq 2^{\frac{2n|\mathcal{Y}_b|^2|\mathcal{Y}_w|^2}{{\ell_q}}}&\max_{s\in\mathcal{S}_{\ell_q}}\bar{P}_e^{(n)},\\
|\mathcal{D}(\mathbb{Q}_{Y_w^n}\parallel \mathbb{Q}_{Z_w^n})-\mathcal{D}(\bar{\mathbb{Q}}_{Y_w^n}\parallel \bar{\mathbb{Q}}_{Z_w^n})|&\leq d_n,
\\
| I(\mathbb{P}_M,\mathbb{V}^n_s)-I(\bar{\mathbb{P}}_M,\bar{\mathbb{W}}^n_s)|&\leq d_n,
\end{aligned}
\end{equation}
where 
\begin{equation*}
    d_n=4n|\mathcal{Y}_b||\mathcal{Y}_w|^2\log{|\mathcal{Y}_w|}/{\ell_q} + 4nH_2\left(|\mathcal{Y}_b||\mathcal{Y}_w|^2/{\ell_q}\right),
\end{equation*}
and $H_2\left(\cdot\right)$ is the binary entropy function. Thus, a $(2^{{\lceil nR\rceil}},n,\epsilon,\delta_{kl},\delta_{s})$-code for $\left(\bar{\mathbb{W}}_s,\bar{\mathbb{V}}_s\right)$ can be used for $\left({\mathbb{W}}_s,{\mathbb{V}}_s\right)$ with a probability of decoding error as follows:
\begin{equation}
\begin{aligned}
\max_{s\in\mathcal{S}}P_e^{(n)}&\leq 2^{\frac{2n|\mathcal{Y}_b|^2|\mathcal{Y}_w|^2}{{\ell_q}}}|\mathcal{S}|^{\frac{1}{4}}2^{-\xi_1\,\omega_n\sqrt{n}}\\&\leq ({\ell_q}+1)^{\frac{|\mathcal{X}||\mathcal{Y}_b||\mathcal{Y}_w|}{4}}2^{-\omega_n\sqrt{n}\left(\xi_1-\frac{2\sqrt{n}|\mathcal{Y}_b|^2|\mathcal{Y}_w|^2}{\omega_n\,{\ell_q}}\right)}.
\end{aligned}
\end{equation}
Therefore, from the first multiplicand, ${\ell_q}$ should not scale faster than $2^{-\omega_n\sqrt{n}}$, i.e., sub-exponentially in $n$, and from the second multiplicand, it should not scale slower than $\frac{\sqrt{n}}{\omega_n}$. We can choose ${\ell_q}$ to scale as a power of $n$ to guarantee the maximum probability of decoding error vanishes when using a code designed for the approximated $\mathfrak{M}_{\ell_q}$. Besides, the KL divergence (and similarly the information leakage) is bounded as follows:
\begin{equation}
\begin{aligned}
\max_{s\in\mathcal{S}}\mathcal{D}(\mathbb{Q}_{Y_w^n}\parallel \mathbb{Q}_{Z_w^n})&\leq\mathcal{D}(\bar{\mathbb{Q}}_{Y_w^n}\parallel \bar{\mathbb{Q}}_{Z_w^n})+d_n\\
&\leq\mathcal{D}_n+2^{-\xi_2\,\omega_n\sqrt{n}}+d_n.
\end{aligned}
\end{equation}
Thus, ${\ell_q}$ should scale faster than $n$ but slower than $2^{-\omega_n\sqrt{n}}$. A good choice can be ${\ell_q}=a\,n^2$, which is the same scaling that is given in \cite{schaefer2015secrecy} where $a>2|\mathcal{Y}_b|^2|\mathcal{Y}_w|^2\, \max{\left\{1, 1/\xi_1\right\}}$. Now, for an arbitrary uncertainty set, Alice can choose a sequence of input distributions such that there exists a sequence of $(2^{{\lceil nR\rceil}},n,\epsilon,\delta_{kl},\delta_{s})$-codes that achieves the following scaling of the compound covert DMC with IS:
\begin{equation}
\begin{aligned}
L_c\geq\varliminf_{n \rightarrow \infty} \sqrt{\frac{n}{2\,\delta_{kl}^2}}&\sup_{\mathbb{P}_n}
& & {\left[\inf_{s\in\mathcal{S}}I(\mathbb{P}_n;\mathbb{W}_s)-\sup_{s\in\mathcal{S}}I(\mathbb{P}_n;\mathbb{V}_s)\right]} \\
& \text{s. t.}
& &\sup_{s\in\mathcal{S}}\mathcal{D}(\mathbb{Q}_{Y_w}\parallel \mathbb{Q}_{Z_w})\leq\frac{2\,\delta_{kl}^2}{n}.
\end{aligned}
\end{equation}

{\bf 3- Existence of a sequence of universal quantizers for all input distributions:} using Lemma 3 and Lemma 4 in \cite{schaefer2015secrecy}, there exists a sequence of universal quantizers $\left\{q_{X,k},q_{Y_b,k},q_{Y_w,k}\right\},\ k\in\mathbb{N},$ for all input distributions and $\forall s\in\mathcal{S}$ and $\mathcal{S}$ is compact such that for every $k>k_\epsilon\in\mathbb{N}$,
\begin{equation}
 I(\bar{\mathbb{P}}_n;\bar{\mathbb{W}}_s)-I(\bar{\mathbb{P}}_n;\bar{\mathbb{V}}_s)\geq\inf_{s\in\mathcal{S}}I(\mathbb{P}_n;\mathbb{W}_s)-\sup_{s\in\mathcal{S}}I(\mathbb{P}_n;\mathbb{V}_s)-\epsilon,
\end{equation}
where $\epsilon>0$. Therefore, the scaling of the compound covert channel with IS, continuous input and output alphabets, and a continuous/compact uncertainty set is bounded as follows:
\begin{equation}
\begin{aligned}
L_c\geq\varliminf_{n \rightarrow \infty} \sqrt{\frac{n}{2\,\delta_{kl}^2}}&\sup_{\mathbb{P}_n}
& & {\left[\inf_{s\in\mathcal{S}}I(\mathbb{P}_n;\mathbb{W}_s)-\sup_{s\in\mathcal{S}}I(\mathbb{P}_n;\mathbb{V}_s)\right]} \\
& \text{s. t.}
& &\sup_{s\in\mathcal{S}}\mathcal{D}(\mathbb{Q}_{Y_w}\parallel \mathbb{Q}_{Z_w})\leq\frac{2\,\delta_{kl}^2}{n}.
\end{aligned}
\end{equation}
Consequently, the scaling of the compound covert MIMO AWGN channels with IS is lower bounded as follows:
\begin{equation}\label{lowerboundcompound}
\begin{aligned}
L_c\geq \varliminf_{n \rightarrow \infty} \sqrt{\frac{n}{2\,\delta_{kl}^2}}R_s,
\end{aligned}
\end{equation}
where 
\begin{equation*}
\begin{aligned}
R_s=
&\max_{\substack{\mathbf{Q}_n \succeq \mathbf{0} \\ 
\mathbf{tr}(\mathbf{Q}_n) \leq P}}\min_{\mathbf{W}_b,\,\mathbf{W}_w\in\mathcal{S}}
  {\left[R_b({\bf Q}_n,\mathbf{W}_b)-R_w({\bf Q}_n,\mathbf{W}_w)\right]} \\
& \text{subject to:}  \max_{\mathbf{W}_w\in\mathcal{S}}\mathcal{D}(\mathbb{P}_{\mathbf{Y}_w}\parallel\mathbb{P}_{\mathbf{Z}_w})\leq\frac{2\,\delta_{kl}^2}{n},
\end{aligned}
\end{equation*}
and $\mathbf{W}_u=\mathbf{H}_u\,\mathbf{H}_u^{ \dag},\ u\in\left\{b,w\right\}$.

{\bf 4- The worst-case scaling (an upper bound on the scaling):} we consider the case when Bob's CSI is known and Willie's CSI is unknown but belongs to the set of channels with a bounded spectral norm. Then, the scaling of the worst-case channel is given as follows:
\begin{equation}
\begin{aligned}
L^{wrst}_c= \varliminf_{n \rightarrow \infty} \sqrt{\frac{n}{2\,\delta_{kl}^2}}R_s^{wrst},
\end{aligned}
\end{equation}
where 
\begin{equation*}
\begin{aligned}
R_s^{wrst}&=\min_{\mathbf{W}_w\in\mathcal{S}}
&\max_{\substack{\mathbf{Q}_n \succeq \mathbf{0} \\ \mathbf{tr}(\mathbf{Q}_n) \leq P}}
&  {\left[R_b({\bf Q}_n)-R_w({\bf Q}_n,\mathbf{W}_w)\right]} \\
& & \text{subject to:}
& \max_{\mathbf{W}_w\in\mathcal{S}}\mathcal{D}(\mathbb{P}_{\mathbf{Y}_w}\parallel\mathbb{P}_{\mathbf{Z}_w})\leq\frac{2\,\delta_{kl}^2}{n}\\
&= &\max_{\substack{\mathbf{Q}_n \succeq \mathbf{0} \\ 
\mathbf{tr}(\mathbf{Q}_n) \leq P}}
&  {\left[R_b({\bf Q}_n)-R_w({\bf Q}_n,\sigma_w^2\,\hat{\lambda}_w{\bf I}_{N_a})\right]} \\
& &\text{subject to:}
& \mathcal{D}_{wrst}(\mathbb{P}_{\mathbf{Y}_w}\parallel\mathbb{P}_{\mathbf{Z}_w})\leq\frac{2\,\delta_{kl}^2}{n},
\end{aligned}
\end{equation*}
and $\mathcal{D}_{wrst}$ is the maximum KL divergence calculated at the worst-case channel $\mathbf{W}_{w}=\sigma_w^2\,\hat{\lambda}_w{\bf I}_{N_a}$.
\begin{IEEEproof}
Since $R_w({\bf Q}_n,\mathbf{W}_w)$ is monotonically increasing in $\mathbf{W}_w$ in the sense that
\begin{equation}
 R_w({\bf Q}_n,\mathbf{W}_{w_1})\leq R_w({\bf Q}_n,\mathbf{W}_{w_2}) \mbox{ if } \mathbf{W}_{w_1}\preceq\mathbf{W}_{w_2}.
\end{equation}
Besides, it holds that $\mathcal{D}(\mathbb{P}_{\mathbf{Y}_{w_1}}\parallel\mathbb{P}_{\mathbf{Z}_w})\leq\mathcal{D}(\mathbb{P}_{\mathbf{Y}_{w_2}}\parallel\mathbb{P}_{\mathbf{Z}_w})$ if $\mathbf{W}_{w_1}\preceq\mathbf{W}_{w_2}$. The reason is that each eigenvalue of $\mathbf{W}_{w_2}$ is no less than each eigenvalue of $\mathbf{W}_{w_1}$, and in the expression of KL divergence, $x-\log(x+1)$ is monotonically increasing in $x,\ \forall x\geq 0$. Thus, the maximum KL divergence, $\mathcal{D}_{wrst}$, in the second equality is obtained at the worst-case channel $\mathbf{W}_{w}=\sigma_w^2\,\hat{\lambda}_w{\bf I}_{N_a}$. Subtracting both sides from $R_b({\bf Q}_n)$ and maximizing subject to the KL constraint and all admissible ${\bf Q}_n$ then, minimizing over all possible Willie's channel realizations completes the proof.
\end{IEEEproof}

{\bf 5- The saddle-point property to connect upper and lower bounds on the scaling:} The following saddle-point property holds.
\begin{equation}
\begin{aligned}
&\min_{\mathbf{W}_w\in\mathcal{S}} &\max_{\substack{\mathbf{Q}_n \succeq \mathbf{0} \\ \mathbf{tr}(\mathbf{Q}_n) \leq P}}
 {\left[R_b({\bf Q}_n)-R_w({\bf Q}_n,\mathbf{W}_w)\right]} \\
& &\text{subject to:}\  \mathcal{D}_{wrst}(\mathbb{P}_{\mathbf{Y}_w}\parallel\mathbb{P}_{\mathbf{Z}_w})\leq\frac{2\,\delta_{kl}^2}{n}\\
=& \max_{\substack{\mathbf{Q}_n \succeq \mathbf{0} \\ \mathbf{tr}(\mathbf{Q}_n) \leq P}}&\min_{\mathbf{W}_w\in\mathcal{S}}
{\left[R_b({\bf Q}_n)-R_w({\bf Q}_n,\mathbf{W}_w)\right]} \\
& \text{subject to:} & \mathcal{D}_{wrst}(\mathbb{P}_{\mathbf{Y}_w}\parallel\mathbb{P}_{\mathbf{Z}_w})\leq\frac{2\,\delta_{kl}^2}{n}.
\end{aligned}
\end{equation}
\begin{IEEEproof} Recall that,
\begin{equation}\begin{aligned}
\min_{\mathbf{W}_w\in\mathcal{S}}&{\left[R_b({\bf Q}_n)-R_w({\bf Q}_n,\mathbf{W}_w)\right]}\\
 &={\left[R_b({\bf Q}_n)-R_w({\bf Q}_n,\sigma_w^2\,\hat{\lambda}_w{\bf I}_{N_a})\right]}.
\end{aligned}\end{equation}
Maximizing both sides yields:
\begin{equation}
\begin{aligned}
& \max_{\substack{\mathbf{Q}_n \succeq \mathbf{0} \\ \mathbf{tr}(\mathbf{Q}_n) \leq P}} \min_{\mathbf{W}_w\in\mathcal{S}} {\left[R_b({\bf Q}_n)-R_w({\bf Q}_n,\mathbf{W}_w)\right]} \\
& \text{subject to:} \  \mathcal{D}_{wrst}(\mathbb{P}_{\mathbf{Y}_w}\parallel\mathbb{P}_{\mathbf{Z}_w})\leq\frac{2\,\delta_{kl}^2}{n}\\
= &\max_{\substack{\mathbf{Q}_n \succeq \mathbf{0} \\ \mathbf{tr}(\mathbf{Q}_n) \leq P}} {\left[R_b({\bf Q}_n)-R_w({\bf Q}_n,\sigma_w^2\,\hat{\lambda}_w{\bf I}_{N_a})\right]} \\
 &\text{subject to:}\   \mathcal{D}_{wrst}(\mathbb{P}_{\mathbf{Y}_w}\parallel\mathbb{P}_{\mathbf{Z}_w})\leq\frac{2\,\delta_{kl}^2}{n}\\
=&\ R^{wrst}_s.
\end{aligned}
\end{equation}
\end{IEEEproof}

{\bf 6- The scaling of the compound covert MIMO AWGN channels:} With known Bob's CSI, the scaling, $L_c$, of the compound covert MIMO AWGN channels equals the scaling, $L^{wrst}_c$, of the worst-case channel and is given as follows:
\begin{equation}
\begin{aligned}
 L_c=\varliminf_{n \rightarrow \infty} \sqrt{\frac{n}{2\delta_{kl}^2}}&\max_{\substack{\mathbf{Q}_n \succeq \mathbf{0} \\ \mathbf{tr}(\mathbf{Q}_n) \leq P}} {\left[R_b({\bf Q}_n)-R_w({\bf Q}_n,\sigma_w^2\hat{\lambda}_w{\bf I}_{N_a})\right]} \\
& \text{subject to:}\  \mathcal{D}_{wrst}(\mathbb{P}_{\mathbf{Y}_w}\parallel\mathbb{P}_{\mathbf{Z}_w})\leq\frac{2\delta_{kl}^2}{n}.
\end{aligned}
\end{equation}
\begin{IEEEproof}
The scaling of the compound channel is upper bounded by the scaling of the worst-case channel and is lower bounded as in \eqref{lowerboundcompound}. Also, by the saddle-point property, both bounds are equal. This completes the proof. 
\end{IEEEproof}


\section{Proof of Theorem \ref{prop:one}}\label{App:proof_prop1}

\begin{IEEEproof} Since Alice transmits in the spatial transmit signature in the directional cosine of Bob's channel, it suffices only to investigate the KL constraint as $N_a$ goes to infinity. The KL constraint is given by:
\begin{equation}\label{LPD}
\log\abs{\frac{1}{\sigma_w^2}\,{\bf Q}_n\,{\bf H}_w ^\dag\,{\bf H}_w+{\bf I}_{N_a}}+\mathbf{tr}\left(\frac{1}{\sigma_w^2}\,{\bf Q}_n\,{\bf H}_w ^\dag\,{\bf H}_w\right)\leq \frac{2\,\delta_{kl}^2}{n}.
\end{equation}

Using the spatial transmit direction of Willie, the KL constraint in \eqref{LPD} can be rewritten as:
\begin{equation}\label{LPD_with_Willie}
 \begin{aligned}
&\log\abs{\frac{1}{\sigma_w^2}\,{\bf Q}_n\,{\bf u}_{\it w}(\Omega_w)\,{\bf u}_{\it w} ^\dag(\Omega_w)+{\bf I}_{N_a}}\\
+&\,\mathbf{tr}\left(\frac{1}{\sigma_w^2}\,{\bf Q}_n\,{\bf u}_{\it w}(\Omega_w)\,{\bf u}_{\it w} ^\dag(\Omega_w)
\right)\leq \frac{2\,\delta_{kl}^2}{n}.
 \end{aligned}
\end{equation}

Using the maximum transmit power at Bob's spatial transmit direction, the KL constraint in \eqref{LPD_with_Willie} can be rewritten as:
\begin{equation}\label{LPD_Bob}
{\frac{P\,{\lambda}_w\abs{f(\Omega)}^2}{\sigma_w^2}}-\log{\left(\frac{P\,{\lambda}_w\abs{f(\Omega)}^2}{\sigma_w^2}+1\right)}\leq\frac{2\,\delta_{kl}^2}{n}, 
\end{equation}
where $\Omega=\Omega_b-\Omega_w$. Letting $N_a$ go to infinity and utilizing the antenna array design for either a fixed array length or a fixed antenna separation as in \eqref{case1}-\eqref{case2}, the left side of the previous equation goes to zero, under the stated conditions, i.e., the KL constraint is satisfied $ \forall \delta_{kl}\geq 0$. 

Further, to estimate the number of transmit antennas that satisfies a predefined target probability of detection, $\delta_{kl}$, we solve the KL constraint in \eqref{LPD_Bob} to obtain $\abs{f(\Omega)}$. The solution of this equation is given in terms of a branch of the Lambert function $W_{-1}$ as follows:
\begin{equation}
\abs{f(\Omega)}^2\leq\frac{\sigma_w^2}{\xi_w^2N_a\,N_w\,P}\left[-W_{-1}\left(-e^{-\frac{2\,\delta_{kl}^2}{n}-1}\right)-1\right],
\end{equation}
where $|{f(\Omega)}|=\abs{\frac{\sin(\pi\,L_a\,\Omega)}{N_a\,\sin(\pi\,\Delta_a\,\Omega)}}$. Hence, 
\begin{equation}
N_a\geq \frac{P\xi_w^2N_w\abs{\sin(\pi L_a\Omega)}^2}{\sigma_w^2\abs{\sin(\pi\,\Delta_a\,\Omega)}^2}\left[-W_{-1}\left(-e^{-\frac{2\delta_{kl}^2}{n}-1}\right)-1\right]^{-1}.
\end{equation} 

Besides, to derive the convergence of the maximal covert coding rate of unit-rank MIMO AWGN channels to the maximal coding rate of unit-rank MIMO AWGN channels, we use the first-order approximation of the maximum covert coding rate to formulate the following problem:
{\small \begin{equation}\label{optimization_problem}
\begin{aligned}
\lim_{N_a\rightarrow \infty}{R}(n,0,\delta_{kl})=\lim_{N_a\rightarrow \infty}&\max_{\substack{\mathbf{Q}_n \succeq \mathbf{0} \\ \mathbf{tr}(\mathbf{Q}_n) \leq P}}
 {\log\abs{\frac{1}{\sigma_{\it b}^2}\,{\bf H}_{\it b}\,{\bf Q}_n\,{\bf H}_{\it b} ^\dag+{\bf I}_{N_{\it b}}}} \\
& \text{subject to:}\  \mathcal{D}(\mathbb{P}_{\mathbf{Y}_w}\parallel\mathbb{P}_{\mathbf{Z}_w})\leq\frac{2\,\delta_{kl}^2}{n},
\end{aligned}
\end{equation} }
The optimization problem is investigated in Appendix \ref{App:power allocation}. Rewrite the gradient of the Lagrangian function in \eqref{eq:primal_optimality} using the spatial transmit signatures as follows:
\begin{equation}\begin{aligned}
\label{eq_transmit_signature}
&\left[{\bf Q}_n^*+\frac{\sigma_{\it b}^2}{\lambda_b}\left({\bf u}_{\it b}(\Omega_b)\,{\bf u}_{\it b} ^\dag(\Omega_b)\right)^{-1} \right]^{-1}\\
+ &\ \eta \left[{\bf Q}_n^*+\frac{\sigma_{\it w}^2}{\lambda_w}\left({\bf u}_{\it w}(\Omega_w)\,{\bf u}_{\it w} ^\dag(\Omega_w)\right)^{-1}\right]^{-1}\\=&\  \mu\,\mathbf{I}_{N_a}+\frac{\eta\,\lambda_w}{\sigma_{\it w}^2}{\bf u}_{\it w}(\Omega_w)\,{\bf u}_{\it w} ^\dag(\Omega_w).
\end{aligned}\end{equation}

With unknown CSI of Willie's channel, Alice transmits in the spatial transmit signature in the directional cosine of Bob's channel, and hence,
\begin{equation}\begin{aligned}
\left[q^*+\frac{\sigma_{\it b}^2}{{\lambda}_b}\right]^{-1}+\eta\left[q^*+\frac{\sigma_{\it w}^2}{{\lambda}_w{\small\abs{f(\Omega)}^2}}\right]^{-1}=\mu+\frac{\eta{\lambda}_w{\small\abs{f(\Omega)}^2}}{\sigma_{\it w}^2}.
\end{aligned}
\end{equation}

Accordingly, we utilize the antenna array design for either a fixed array length or a fixed antenna separation as in \eqref{case1}-\eqref{case2}, under the stated conditions, $\left[q^*+\frac{\sigma_{\it b}^2}{{\lambda}_b}\right]^{-1}=\mu$ as $N_a$ goes to infinity. Hence, the maximal covert coding rate of unit-rank MIMO AWGN channels converges to the maximal coding rate of unit-rank MIMO AWGN channels and the KL constraint is satisfied, i.e., $\lim_{N_a\rightarrow \infty}{R}(n,0,\delta_{kl})=\lim_{N_a\rightarrow \infty}{R}(n,0)$, for any given $\delta_{kl}\geq 0$ and under the stated conditions. 
\end{IEEEproof}

\section{Proof of Theorem \ref{prop:two}}\label{App:proof_prop2}

\begin{IEEEproof}Using the angular domain representation, the gradient of the Lagrangian function in \eqref{eq:primal_optimality} can be rewritten as follows: 
\begin{equation}\begin{aligned}\label{eq:asymptotic_decomposition}
\left[\tilde{\bf Q}_n^*+\sigma_{\it b}^2\left(\mathbf{H}_b^{g\dag}\,\mathbf{H}_b^{g}\right)^{-1} \right]^{-1}+\eta&\left[\tilde{\bf Q}_n^*+\sigma_{\it w}^2\left({\bf H}_{\it w}^{g \dag}\,{\bf H}^g_{\it w}\right)^{-1} \right]^{-1}\\&=\mu\,\mathbf{I}_{N_a}+ \frac{\eta}{\sigma_{\it w}^2}{\bf H}_{\it w}^{g \dag}\,{\bf H}^g_{\it w},
\end{aligned}\end{equation}
where $\tilde{\bf Q}_n^*={\bf U}_{\it b}^\dag\,{\bf Q}_n^*\,{\bf U}_{\it b}$ is a diagonal matrix with diagonal elements ${q}_i^*$ and $\bf{tr}({\tilde{\mathbf{Q}}_n^*})=\bf{tr}({\mathbf{Q}_n^*})$, and ${\bf U}_{\it b}$ is a unitary transmit matrix whose columns are the orthonormal basis of the transmitted signal space. Hence, the angular domain representation of the gradient in \eqref{eq:asymptotic_decomposition} can be decomposed in each angular window of Bob's channel as follows, 
\begin{equation}
\left[{q}^*_i+\frac{\sigma_{\it b}^2}{{\lambda}_{b,i}}\right]^{-1}+ \eta \left[{q}^*_i+\frac{\sigma_{\it w}^2}{\tilde{\lambda}_{w,i}}\right]^{-1}= \mu +\frac{\eta\,\tilde{\lambda}_{\it w,i}}{\sigma_{\it w}^2},
\end{equation}
$\forall i\in\left\{1,\dots,N_a\right\}$, where
{\small\begin{equation}
\begin{aligned}
\tilde{\lambda}_{w,i}&=\mathbf{u}^{\dag}_{b,i}(\Omega_{b,j})\left[\sum_{j=1}^{N_a}\lambda_{w,j}\mathbf{u}_{w,j}(\Omega_{w,j})\mathbf{u}^{\dag}_{w,j}(\Omega_{w,j})\right]\mathbf{u}_{b,j}(\Omega_{b,i})\\
&=\sum_{j=1}^{N_a}\lambda_{w,j}\,\abs{f(\Omega_{i,j})}^2,
\end{aligned} 
\end{equation}}
and $\Omega_{i,j}=\Omega_{b,i}-\Omega_{w,j}$. When Alice transmits in the spatial transmit signatures in the directional cosines of Bob's channel, we utilize the antenna array design for either a fixed array length or a fixed antenna separation as in \eqref{case1}-\eqref{case2}, $ \forall i,$ $\left[{q}^*_i+\frac{\sigma_{\it b}^2}{{\lambda}_{b,i}}\right]^{-1}=\mu$, as $N_a$ goes to infinity. Thus, the maximal covert coding rate of multi-path MIMO AWGN channels converges to the maximal coding rate of multi-path MIMO AWGN channels and the KL constraint is satisfied, i.e., $\lim_{N_a\rightarrow \infty}{R}(n,0,\delta_{kl})=\lim_{N_a\rightarrow \infty}{R}(n,0)$, for any given $\delta_{kl}\geq 0$ and under the stated conditions.
\end{IEEEproof}



\end{document}